\documentclass[a4paper,11pt,reqno]{amsart}
\usepackage[a4paper]{geometry}
\usepackage{amssymb,amsmath}
\usepackage{mathtools}
\usepackage{enumerate}
\usepackage{bookmark}
\usepackage{hyperref}
\usepackage[initials,backrefs]{amsrefs}

\numberwithin{equation}{section}
\theoremstyle{plain}
\newtheorem{theorem}{Theorem}

\newtheorem{lemma}{Lemma}

\theoremstyle{definition}
\newtheorem{definition}{Definition}
\newtheorem*{assi*}{(I) Short-range interaction}
\newtheorem*{assp*}{(P) Log-H\"older continuity condition}
\newtheorem*{dskn*}{$\dskn$}
\newtheorem*{dsknn*}{$\dsknn$}
\theoremstyle{remark}

\newcommand{\prob}[1]{\DP\left\{#1\right\}}
\newcommand{\esm}[1]{\mathbb{E}\left[\,#1\,\right]}
\newcommand{\Bone}{\mathbf{1}}
\newcommand{\BA}{\mathbf{A}}
\newcommand{\BB}{\mathbf{B}}
\newcommand{\BC}{\mathbf{C}}

\newcommand{\BG}{\mathbf{G}}
\newcommand{\BH}{\mathbf{H}}

\newcommand{\BK}{\mathbf{K}}

\newcommand{\BP}{\mathbf{P}}
\newcommand{\BU}{\mathbf{U}}

\newcommand{\BV}{\mathbf{V}}

\newcommand{\BX}{\mathbf{X}}

\newcommand{\CJ}{\mathcal{J}}

\newcommand{\DN}{\mathbb{N}}
\newcommand{\DP}{\mathbb{P}}
\newcommand{\DR}{\mathbb{R}}
\newcommand{\DZ}{\mathbb{Z}}
\newcommand{\BDelta}{\mathbf{\Delta}}
\newcommand{\BPsi}{\mathbf{\Psi}}
\newcommand{\Bx}{\mathbf{x}}
\newcommand{\By}{\mathbf{y}}

\newcommand{\Bu}{\mathbf{u}}
\newcommand{\Bv}{\mathbf{v}}
\newcommand{\FB}{\mathfrak{B}}

\newcommand{\rA}{\mathrm{A}}
\newcommand{\rB}{\mathrm{B}}

\newcommand{\rR}{\mathrm{R}}
\newcommand{\rS}{\mathrm{S}}
\newcommand{\rT}{\mathrm{T}}
\DeclareMathOperator{\card}{card}
\DeclareMathOperator{\diam}{diam}
\DeclareMathOperator{\dist}{dist}
\DeclareMathOperator{\supp}{supp}

\newcommand{\ee}{\mathrm{e}}

\newcommand{\comp}{\mathrm{c}}
\newcommand{\fui}{\mathrm{FI}}

\newcommand{\pai}{\mathrm{PI}}
\newcommand{\sep}{\mathrm{sep}}

\newcommand{\condI}{\mathbf{(I)}}
\newcommand{\condP}{\mathbf{(P)}}

\newcommand{\dsn}[2]{\mathbf{(DS.}#1#2,$\,N\mathbf{)}$}

\newcommand{\dskn}{\mathbf{(DS.}k,N\mathbf{)}}
\newcommand{\dsknn}{\mathbf{(DS.}k,n,N\mathbf{)}}

\newcommand{\dskonn}{\mathbf{(DS.}k+1,n,N\mathbf{)}}
\newcommand{\tto}[1]{\smash{\mathop{\,\,\,\, \longrightarrow \,\,\,\, }\limits_{#1}}}

\begin{document}
\title[localization at low energy for multi-particle continuous models]{Multi-particle localization at low energy for the multi-dimensional continuous Anderson model}

\author[T.~Ekanga]{Tr\'esor EKANGA$^{\ast}$}

\address{$^{\ast}$%
Institut de Math\'ematiques de Jussieu, 
Universit\'e Paris Diderot,
Batiment Sophie Germain,
13 rue Albert Einstein,
75013 Paris,
France}
\email{tresor.ekanga@imj-prg.fr}
\date{\today}
\begin{abstract}
We study the multi-particle Anderson model in the continuum and show that under some mild assumptions on the random external potential and the inter-particle interaction, for any finite number of particles, the multi-particle lower edges of the spectrum are almost surely constant in absence of ergodicity. We stress that this result is not quite obvious and  has to be handled carefully.  In addition, we prove the spectral exponential and the strong dynamical localization of the continuous multi-particle Anderson model at low energy. The proof based on the multi-particle multi-scale analysis bounds, needs the values of the  external random potential to be independent and identically distributed (i.i.d.) whose common probability distribution is at least Log-H\"older continuous. 
\end{abstract}
\maketitle

\section{Introduction}
This paper follows our previous works \cites{E11,E13} on localization for multi-particle random lattice Schr\"odinger operators at low energy. Some other papers \cites{CS09a,CS09b,BCSS10a,AW09,AW10,FW15,KN13,KN14,C14} analyzed multi-particle models in the regime including the strong disorder or the low energy and for different types of models such as the alloy-type Anderson model or the multi-particle Anderson model in quantum graphs \cite{S14}.  

In their work \cite{KN14}, Klein and Nguyen developed the continuum multi-particle bootstrap multi-scale analysis of the Anderson model with alloy type external potential. The method of Klein and Nguyen is very closed in the spirit to that of our work \cite{E13} although this work is not mentioned at all. The results of \cite{E13} was the first rigorous mathematical proof of localization for many body interacting Hamiltonians near the bottom of the spectrum on the lattice and in the present paper, we prove the similar results in the continuum.

The work by Sabri \cite{S14}, uses a different strategy in the course of the multi-particle multi-scale analysis at low energy. The analysis is made by considering the Green's functions, i.e., the matrix elements of the local resolvent operator instead of the norm of the kernel operators as it will be developed in this paper and this obliged the author to modify the standard Combes Thomas estimate and adapt it to the matrix elements of the local resolvents. Also, our proof on the almost surely spectrum is completely different and the scale induction step in the multi-particle multi-scale analysis as well as the strategy of the proof of the localization results. Chulaevsky himself \cite{C14}, used the results of Klein and Nguyen \cite{KN14} and analyzed multi-particle random operators with alloy-type external potential with infinite range interaction at low energy.

Let us emphasize that the almost sure non-randomness of the bottom of the spectrum of the  multi-particle random Hamiltonian is the heart of the problem of localization at low energy for multi-particle systems. In this work, we propose a very clear and constructive proof of the almost surely non-randomness of the multi-particle lower spectral edges and both the exponential localization of the eigenfunctions in the max-norm and the strong dynamical localization near the bottom of the spectrum.

Our multi-particle multi-scale analysis is more closed in the spirit to its single-particle counterpart developed by Stollmann \cite{St01}, in the continuum case and by Dreifus and Klein \cite{DK89} in the lattice case. The multi-particle multi-scale analysis bounds are proved for energies in a fix interval of the form $(-\infty; E^*]$ while in \cite{S14}, the author restricted his analysis to compact intervals of the form $[nq_{-}; nq_{+}]$ depending on the number of particles $n\geq1$ and this results in some complications in the scaling analysis dealing with perturbed energies $E-\lambda_i$ for energies $E$ belonging  to $[nq_{-};nq_{+}]$ because the perturbed energy $E-\lambda_i$ might be out of this interval. In that case, the author applied again some Combes Thomas bound. We encountered this problem by proving first that all the lower spectral edges are equal and second that the initial length scale estimates are valid in some unbounded from below intervals and the problem  is resolved using the non-negativity hypothesis on the external random potential and the interaction potential.

Let us now discuss on the structure of the paper. In the next Section, we set up the model, give the assumptions and formulate the main results.  In Section 3, We give the two important results for our multi-scale analysis scheme. Namely, the Wegner and the Combes Thomas estimates. The one important to bound in probability the resonances and the other useful in the study of initial length scales estimates of the multi-scale analysis at low energy. Section 4 is devoted to the multi-particle multi-scale induction step using the assumption on the interaction potential. Finally, in Section 5, we prove the main results.

\section{The model, hypotheses and the main results}

\subsection{The model}
We fix at the very beginning the number of particles $N\geq 2$. We are concern with multi-particle random Schr\"odinger operators of the following forms:
\[
\BH^{(N)}(\omega):=-\BDelta + \BU+\BV,
\]
acting in $L^{2}((\DR^{d})^N)$. Sometimes, we will use the identification $(\DR^{d})^N\cong \DR^{Nd}$. Above, $\BDelta$ is the Laplacian on $\DR^{Nd}$, $\BU$ represents the inter-particle interaction which acts as multiplication operator in $L^{2}(\DR^{Nd})$. Additional information on $\BU$ is given in the assumptions. $\BV$ is the multi-particle random external potential also acting as multiplication operator on $L^{2}(\DR^{Nd})$. For $\Bx=(x_1,\ldots,x_N)\in(\DR^{d})^N$, $\BV(\Bx)=V(x_1)+\cdots+ V(x_N)$ and $\{V(x,\omega), x\in\DR^d\}$ is a random i.i.d. stochastic process relative to the  probability space $(\Omega,\FB,\DP)$ with $\Omega=\DR^{\DZ^d}$, $\FB=\bigotimes_{x\in\DZ^d}B(\DR)$ and $\DP=\bigotimes_{x\in\DZ^d}\mu$ where $\mu$ is the common probability distribution measure of the i.i.d. variables $\{V(x,\omega)\}$.

Observe that the non-interacting Hamiltonian $\BH^{(N)}_0(\omega)$ can be written as a tensor product:
\[
\BH^{(N)}_0(\omega):=-\BDelta +\BV=\sum_{k=1}^N \Bone^{\otimes(k-1)}_{L^{2}(\DR^d)}\otimes H^{(1)}(\omega)\otimes \Bone^{\otimes(N-k)}_{L^2(\DR^d)},
\]
where, $H^{(1)}(\omega)=-\Delta + V(x,\omega)$ acting on $L^2(\DR^d)$. We will also consider random Hamiltonian $\BH^{(n)}(\omega)$, $n=1,\ldots,N$ defined similarly. Denote by $|\cdot|$ the max-norm in $\DR^{nd}$.

\subsection{Assumptions}

\begin{assi*}

Fix any $n=1,\ldots,N$. The potential of inter-particle interaction $\mathbf{U}$ is bounded, non-negative and of the form
\[
\BU(\Bx)=\sum_{1\leq i<j\leq n}\Phi(|x_i-x_j|),\quad \Bx=(x_1,\ldots,x_n),
\]
where  $\Phi:\DN:\rightarrow\DR$ is a compactly supported function such that

\begin{equation}\label{eq:finite.range.k}
\exists r_0\in\DN: \supp \Phi\subset[0,r_0].
\end{equation}

\end{assi*}

The external random potential $V\colon\DZ^d\times\Omega\to\DR$ is an i.i.d. random field
relative to  $(\Omega,\FB,\DP)$ and is defined by $V(x,\omega)=\omega_x$ for $\omega=(\omega_i)_{i\in\DZ^d}$.
The common probability distribution function, $F_V$,  of the i.i.d. random variables $V(x,\cdot)$, $x\in\DZ^d$ associated to the measure $\mu$ is defined by
\[
F_V: t \mapsto \prob{V(0,\omega)\leq t }.
\]
\begin{assp*}
The random potential field $\{V(x,\omega)\}_\{x\in \DZ^d\}$ is i.i.d. , of non-negative values and the corresponding probability distribution function $F_V$ is log-H\"older continuous:  More precisely,
\begin{align}\label{eq:assumption.3prime}
&s(F_V,\varepsilon) := \sup_{a\in\DR}(F_V(a+\varepsilon)-F_V(a))
\leq\frac{C}{|\ln\epsilon|^{2A}}\\
&\text{for some }C\in(0,\infty)\text{ and }A>\frac{3}{2}\times4^Np+9Nd.\notag
\end{align}
Note that this last condition depends on the parameter $p$ which will be introduced  in Section \ref{sec:Nparticle.scheme}.
\end{assp*}

\subsection{The results}\label{sec:main.results}

For any $n=1,\ldots,N$, we denote by $\sigma(\BH^{(n)}(\omega))$ the spectrum of $\BH^{(n)}(\omega)$ and $E^{(n)}_0(\omega)$ the infimum of $\sigma(\BH^{(n)}(\omega))$.

\begin{theorem}\label{thm:bottom.spectrum}
Let $1\leq n\leq N$. Under  assumptions  $\condI$ and $\condP$, we have with probability one:
\[
\sigma(\BH^{(n)}(\omega))=[0,+\infty).
\]
Consequently,
\[
E^{(n)}_0:= \inf\sigma(\BH^{(n)}(\omega))=0 \quad \emph{a.s.}
\]
\end{theorem}

\begin{theorem}\label{thm:exp.loc}
Under the assumptions $\condI$ and $\condP$, there exists $E^*>E_0^{(N)}$ such that with $\DP$-probability one,
\begin{enumerate}[\rm(i)]
\item
the spectrum of $\BH^{(N)}(\omega)$ in $[E_0^{(N)},E^*]$ is nonempty and pure point;
\item
any eigenfunction corresponding to an eigenvalue in $[E_0^{(N)},E^*]$ is exponentially decaying at infinity in the max-norm.
\end{enumerate}
\end{theorem}

\begin{theorem}\label{thm:dynamical.loc}
Assume  that the hypotheses   $\condI$ and $\condP$ are valid, there exist $E^*>E^{(N)}_0$ and $s^*(N,d)>0$ such that for any bounded $\BK\subset \DZ^{Nd}$ and any $0<s<s^*$ we have 
\begin{equation}\label{eq:low.energy.dynamical.loc}
\esm{\sup_{t>0}\Bigl\| |\BX|^{\frac{s}{2}}\ee^{-it\BH^{(N)}(\omega)}\BP_{I}(\BH^{(N)}(\omega))\Bone_{\BK}\Bigr\|_{L^2(\DR^{Nd})}}<\infty,
\end{equation}
where $(|\BX|\BPsi)(\Bx):=|\Bx|\BPsi(\Bx)$, $\BP_{I}(\BH^{(N)}(\omega))$ is the spectral projection of $\BH^{(N)}(\omega)$ onto the interval $I:=[E^{(N)}_0,E^*]$ and $\BK\subset\DR^{Nd}$ is a compact domain.
\end{theorem}

Some parts of the rest of the text overlap with that of the paper \cite{E16}, but for the reader convenience we give all the details of the arguments.

\section{Input for the multi-particle multi-scale analysis and geometry}\label{sec:Nparticle.scheme}

\subsection{Geometric facts}
According to the general structure of the MSA, we work with  \emph{rectangular} domains. For $\Bu=(u_1,\ldots,u_n)\in\DZ^{nd}$, we denote by $\BC^{(n)}_L(\Bu)$ the $n$-particle open cube, i.e,
\[
\BC^{(n)}_L(\Bu)=\left\{\Bx\in\DR^{nd}:|\Bx-\Bu|< L\right\},
\]
and given $\{L_i: i=1,\ldots,n\}$, we define the rectangle
\begin{equation}          \label{eq:cube}
\BC^{(n)}(\Bu)=\prod_{i=1}^n C^{(1)}_{L_i}(u_i),
\end{equation}
where $C^{(1)}_{L_i}(u_i)$ are cubes of side length $2L_i$ center at points $u_i\in\DZ^d$. We also define 
\[
\BC^{(n,int)}_L(\Bu):=\BC^{(n)}_{L/3}(\Bu), \quad \BC^{(n,out)}_L(\Bu):=\BC^{(n)}_L(\Bu)\setminus\BC^{(n)}_{L-2}(\Bu), \quad \Bu\in\DZ^{nd}
\]
and introduce  the characteristic functions:
\[
\Bone^{(n,int)}_{\Bx}:=\Bone_{\BC^{(n,int)}_L(\Bx)}, \qquad \Bone^{(n,out)}_{\Bx}:= \Bone_{\BC^{(n,out)}_L(\Bx)}.
\]
The volume of the cube $\BC^{(n)}_L(\Bu)$ is $|\BC_L^{(n)}(\Bu)| :=(2L)^{nd}$.
We denote the restriction of the Hamiltonian $\BH_h^{(n)}$ to  $\BC^{(n)}(\Bu)$ by
\begin{align*}
&\BH_{\BC^{(n)}(\Bu),h}^{(n)}=\BH^{(n)}_h\big\vert_{\BC^{(n)}(\Bu)}\\
&\text{with Dirichlet boundary conditions}
\end{align*}

We denote the spectrum of $\BH_{\BC^{(n)}(\Bu),h}^{(n)}$  by
$\sigma\bigl(\BH_{\BC^{(n)}(\Bu)}^{(n),h}\bigr)$ and its resolvent by
\begin{equation}\label{eq:def.resolvent}
\BG^{(n)}_{\BC^{(n)}(\Bu),h}(E):=\Bigl(\BH_{\BC^{(n)}(\Bu),h}^{(n)}-E\Bigr)^{-1},\quad E\in\DR\setminus\sigma\Bigl(\BH_{\BC^{(n)}(\Bu),h}^{(n)}\Bigr).
\end{equation}

Let $m>0$ and $E\in\DR$ be given.  A cube $\BC_L^{(n)}(\Bu)\subset\DR^{nd}$, $1\leq n\leq N$ will be  called $(E,m,h)$-\emph{nonsingular} ($(E,m,h)$-NS) if $E\notin\sigma(\BH^{(n)}_{\BC^{(n)}_{L}(\Bu),h})$ and
\begin{equation}\label{eq:singular}
\|\Bone^{(n,out)}_{\Bx}\BG^{(n)}_{\BC^{(n)}_L(\Bx)}(E)\Bone^{(n,int)}_{\Bx}\|\leq\ee^{-\gamma(m,L,n)L},
\end{equation}
where
\begin{equation}\label{eq:gamma}
\gamma(m,L,n)=m(1+L^{-1/8})^{N-n+1}.                     
\end{equation}
Otherwise it will be called $(E,m,h)$-\emph{singular} ($(E,m,h)$-S).

Let us introduce the following.
\begin{definition}
Let $n\geq 1$, $E\in\DR$ and $\alpha=3/2$.  
\begin{enumerate}[\rm(A)]
\item
A cube $\BC_L^{(n)}(\Bv)\subset\DR^{nd}$  is called $(E,h)$-resonant ($(E,h)$-R) if
\begin{equation} \label{eq:E-resonant}
\dist\Bigl[E,\sigma\bigl(\BH_{\BC_L^{(n)}(\Bv),h}^{(n)}\bigr)\Bigr]\leq\ee^{-L^{1/2}}.
\end{equation}
Otherwise it is called $(E,h)$-non-resonant ($(E,h)$-NR).
\item
A cube $\BC_L^{(n)}(\Bv)\subset\DR^{nd}$ is called $(E,h)$-completely nonresonant ($(E,h)$-CNR), if it does not contain any $(E,h)$-R cube of size $\geq L^{1/\alpha}$. In particular $\BC^{(n)}_L(\Bv)$ is itself $(E,h)$-NR.
\end{enumerate}
\end{definition}

We will also make use of the following notion.

\begin{definition}\label{def:separability}
A cube $\BC^{(n)}_L(\Bx)$ is  $\CJ$-separable from $\BC^{(n)}_L(\By)$ if there exists a nonempty subset $\CJ\subset\{1,\cdots,n\}$ such that
\[
\left(\bigcup_{j\in \CJ}C^{(1)}_{L}(x_j)\right)\cap
\left(\bigcup_{j\notin \CJ}C_L^{(1)}(x_j)\cup \bigcup_{j=1}^n C_L^{(1)}(y_j)\right)=\emptyset.
\]
A pair $(\BC^{(n)}_L(\Bx),\BC^{(n)}_L(\By))$ is separable if $|\Bx-\By|>7NL$ and if one of the cube is $\CJ$-separable from the other.
\end{definition}
\begin{lemma}\label{lem:separable.distant}
Let $L>1$.
\begin{enumerate}[\rm(A)]
\item
For any $\Bx\in\DZ^{nd}$, there exists a collection of $n$-particle cubes
$\BC^{(n)}_{2nL}(\Bx^{(\ell)})$ with $\ell=1,\ldots,\kappa(n)$, $\kappa(n)= n^n$, $\Bx^{\ell}\in\DZ^{nd}$ such that if
 $\By\in\DZ^{nd}$ satisfies $|\By-\Bx|>7NL$ and
\[
\By \notin \bigcup_{\ell=1}^{\kappa(n)} \BC^{(n)}_{2nL}(\Bx^{(\ell)})
\]
then the cubes $\BC^{(n)}_L(\Bx)$ and $\BC^{(n)}_L(\By)$ are separable.
\item
Let $\BC^{(n)}_L(\By)\subset \DR^{nd}$ be an $n$-particle cube. Any cube  $\BC^{(n)}_L(\Bx)$ with 
\[
|\By-\Bx|>\max_{1\leq i,j\leq n}|y_i-y_j| +5NL,
\]
 is $\CJ$-separable from
$\BC^{(n)}_L(\By)$ for some $\CJ\subset\{1,\ldots,n\}$.
\end{enumerate}
\end{lemma}
\begin{proof}
See the appendix section \ref{sec:appendix}.
\end{proof}

\subsection{The multi-particle Wegner estimates}
In our earlier works \cites{E11,E13,E16}  as well as in other previous papers in the multi-particle localization theory \cites{CS09b,BCSS10b} the notion of separability was crucial in order to prove the Wegner estimates for pairs of multi-particle cubes via the Stollmann's Lemma. It is plain (cf. \cite{E13}, Section 4.1), that  sufficiently distant pairs of fully interactive cubes have disjoint projections and this fact combined with independence is used in that case to bound the probability of an intersection of events relative to those projections. We state below the Wegner estimates directly in a form suitable for our multi-particle multi-scale analysis using assumption $\condP$.

\begin{theorem}\label{thm:Wegner}
 Assume that the random potential satisfies assumption $\condP$, then 
\begin{enumerate}

\item[\rm(A)]
for any $E\in\DR$
\begin{equation}\label{eq:cor.Wegner.2A}
\prob{\text{$ \BC^{(n)}_L(\Bx)$ is not $E$-CNR }}\leq L^{-p\,4^{N-n}},
\end{equation}

\item[\rm(B)]
\begin{equation}\label{eq:cor.Wegner.2B}
\prob{\text{$\exists E\in \DR:$ neither $\BC^{(n)}_{L}(\Bx)$ nor $\BC^{(n)}_{L}(\By)$ is $E$-CNR}} \leq L^{-p\,4^{N-n}},
\end{equation}
\end{enumerate}
where $p>6Nd$, depends only on the fixed number of particles $N$ and the configuration dimension $d$.
\end{theorem}

\begin{proof}
See the articles \cites{BCSS10b,CS08}.
\end{proof}

We also give the Combes Thomas estimates in 

\begin{theorem}\label{thm:CT}
Let $H=-\BDelta + W$ be a Schr\"odinger operator on $L^2(\DR^{D})$, $E\in\DR$ and $E_0=\inf \sigma(H)$. Set $\eta=\dist(E,\sigma(H))$. If $E<E_0$,, then for any $0<\gamma<1$, we have that:
\[
\left\|\Bone_{x}(H-E)^{-1}\Bone_{y}\right\|\leq \frac{1}{(1-\gamma^2)\eta}\ee^{\gamma\sqrt{\eta d}}\ee^{-\gamma \sqrt{\eta}|x-y|},
\]
for all $x,y\in\DR^{D}$. 
\end{theorem}  

\begin{proof}
See the proof of Theorem $1$ in \cite{GK02}.
\end{proof}

We define the mass $m>0$ depending on the parameters $N$, $\gamma$ and the initial length scale $L$ in the following way:

\begin{equation}\label{eq:m}
m:=\frac{2^{-N}\gamma L^{-1/4}}{3\sqrt{2}}.
\end{equation}

We recall below the geometric resolvent inequality and the eigenfunction decay inequality.

\begin{theorem}[Geometric resolvent inequality (GRI)]\label{thm:GRI.GF}
For a given bounded interval $I_0\subset\DR$, there is a constant $C_{geom}>0$ such that for $\BC^{(n)}_{\ell}(\Bx)\subset\BC^{(n)}_L(\Bu)$, $\BA\subset\BC^{(n,int)}_{\ell}(\Bx)$, $\BB\subset\BC^{(n)}_L(\Bu)\setminus\BC^{(n)}_{\ell}(\Bx)$ and $E\in I_0$, the following inequality holds true:
\[
\|\Bone_{\BB}\BG^{(n)}_{\BC^{(n)}_L(\Bu)}(E)\Bone_{\BA}\|\leq C_{geom}\cdot\|\Bone_{\BB}\BG^{(n)}_{\BC^{(n)}_L(\Bu)}(E)\Bone_{\BC^{(n,out)}_{\ell}(\Bx)}\|\cdot\|\Bone_{\BC^{(n,out)}_{\ell}(\Bx)}\BG^{(n)}_{\BC^{(n)}_{\ell}(\Bx)}(E)\Bone_{\BA}\|.
\]
\end{theorem}
\begin{proof}
See \cite{St01}, Lemma 2.5.4.
\end{proof}

\begin{theorem}[Eigenfunctions decay inequality (EDI)] \label{thm:GRI.EF}
For every $E\in\DR$, $\BC^{(n)}_{\ell}(\Bx)\subset \DR^{nd}$ and every polynomially bounded function $\BPsi\in L^2(\DR^{nd})$:
\[
\|\Bone_{\BC^{(n)}_1(\Bx)}\cdot\BPsi\|\leq C\cdot\|\Bone_{\BC^{(n,out)}_{\ell}(\Bx)}\BG^{(n)}_{\BC^{(n)}_{\ell}(\Bx)}(E)\Bone_{\BC^{(n,int)}_{\ell}(\Bx)}\|\cdot\|\Bone_{\BC^{(n,out)}_{\ell}(\Bx)}\cdot\BPsi\|.
\]
\end{theorem}
\begin{proof}
See section 2.5 and proposition 3.3.1 in \cite{St01}.
\end{proof}

\section{The initial bounds of the multi-particle multi-scale analysis }
 In this Section, we denote by $E_0^{(n)}(\omega)$ the bottom of the spectrum of the Hamiltonian $\BH^{(n)}_{\BC^{(n)}_L(\Bu)}(\omega)$, i.e., $E_0^{(n)}(\omega):=\inf\sigma(BH^{(n)}_{\BC^{(n)}_L(\Bu)}(\omega)$.
 We the following bound from the single-particle localization theory,
 
\begin{theorem}\label{thm:1p.loc}
Under the hypotheses $\condI$ and $\condP$, for any $p>0$, there exists $L^*>0$ such that 
\[
\prob{E_0^{(1)}(\omega)\leq L^{-1/2}} \leq L^{2p4^{N-1}},
\]
for all $L\geq L^*$.
\end{theorem}
\begin{proof}
See the book by Peter Stollmann Section 2.
\end{proof}

Now, we show that the same result holds true for the multi-particle random Hamiltonian in the following statement.

\begin{theorem}\label{thm:np.bottom}

Under the hypotheses $\condI$ and $\condP$, for any $p>0$, there exists $L^*_1>0$ such that 
\[
\prob{E_0^{(n)}(\omega)\leq L^{-1/2}} \leq L^{2p4^{N-n}},
\]
for all $L\geq L^*_1$.
\end{theorem}

\begin{proof}
We denote by $\BH^{(n)}_0(\omega)$ the multi-particle random Hamiltonian without interaction. Observe that, since the interaction potential $\BU$ is non-negative, we have 
\[
E_0(\BH^{(n)}_{\BC^{(n)}_L(\Bu)}(\omega)\geq  E_0^{(n)}(\omega),
\]
where $E_0^{(n)}(\omega)=\lambda_1^{(1)}(\omega)+\cdots+\lambda_n^{(1)}$ and the $\lambda_i^{(1)}(\omega)$ are the eigenvalues of the single-particle random Hamiltonians $H^{(1)}_{\BC^{(1)}_L(u_i)}(\omega)$, $i=1,\ldots,n$. So, if $E^{(n)}_0(\omega)\leq L^{-1/2}$, then for example $\lambda_1^{(1)}(\omega)\leq L^{-1/2}$ and this implies the required probability bound of the assertion of Theorem \ref{thm:np.initial.MSA}.
\end{proof}

We are now ready to prove our initial length scale estimate of the multi-particle multi-scale analysis given below.  

Recall that the parameter $m>0$ is given by $m=\frac{2^{-N}\gamma L^{-1/4}}{3\sqrt{2}}$.
\begin{theorem}\label{thm:np.initial.MSA}

Assume that the hypotheses $\condI$ and $\condP$ hold true. Then, there exists $E^*>0$ such that 
\[
\prob{\text{$\exists E\in (-\infty: E^*]:$ $\BC^{(n)}_L(\Bu)$ is $(E,m)$-S}}\leq L^{-2p4^{N-n}},
\]
for $L>0$ large enough.
\end{theorem}

\begin{proof}
Set $E^*:=\frac{1}{2} L_0^{-1/2}$. If the first eigenvalue $E_0^{(n)}(\omega)$ satisfies $E_0^{(n)}(\omega)>L^{-1/2}$, then for all energy $E\leq E^*$, we have:

\begin{align*}
\dist(E,\sigma(\BH^{(n)}_{\BC^{(n)}_L(\Bu)}))&=E^{(n)}_0(\omega)\\
&> L^{-1/2}-\frac{1}{2} L^{-1/2}\\
&>\frac{1}{2} L^{-1/2}
\end{align*}
Thus by the Combes Thomas estimate Theorem \ref{thm:CT},
\begin{align*}
\|\Bone_{\Bx}\BG^{(n)}_{\BC^{(n)}_L(\Bu)}(E)\Bone_{\By}\|&\leq 2L^{-1/2}\ee^{\gamma \sqrt{d}\sqrt{\eta}}\ee^{-\gamma \sqrt{\eta} |\Bx-\By|}\\
&\leq 2 L^{1/2}\ee^{-\frac{\gamma L^{-1/4}}{\sqrt{2}}(\frac{L}{3}-\sqrt{d})}
\end{align*}
 
Thus for $L>0$ large enough depending on the dimension $d$, we get
\begin{gather*}
\|\Bone_{\BC^{(n,out)}_L(\Bu)}\BG^{(n)}_{\BC^{(n)}_L(\Bu)}(E)\Bone_{\BC^{(n,int)}_L(\Bu)}\|\\
\leq \sum\limits_{{\substack{\Bx\in\BC^{(n,out)}_L(\Bu)\cap_\DZ^{nd}\\ \By\in\BC^{(n,int)}_L(\Bu)\cap\DZ^{nd}}}} 2L^{-1/2}\ee^{-\frac{\gamma L^{-1/4}}{\sqrt{2}}(\frac{L}{3}-\sqrt{d})}\\
\leq (2L)^{2nd} 2 L^{1/2}\ee^{-2^N m L}\\
\end{gather*}

Now since $\gamma(m,L,n)=m(1L^{-1/8})^{N-n}<2^N m$, for $L>0$, large enough, we have that
\[
\|\Bone_{\BC^{(n,out)}_L(\Bu)}\BG^{(n)}_{\BC^{(n)}_L(\Bu)}(E)\Bone_{\BC^{(n,int)}_L(\Bu)}\|\leq \ee^{-\gamma(m,L,n)L}.
\]
The above analysis, then implies that

\begin{gather*}
\prob{\text{$\exists E\leq E^*$: $\BC^{(n)}_L(\Bu)$ is $(E,m)$-S}}\\
\leq \prob{E^{(n)}_0(\omega)\leq L^{-1/2}}\leq L^{-2p 4^{N-n}},
\end{gather*}
Yielding the required result.

\end{proof}

Below, we develop the induction step of the multi-scale analysis and although the text overlaps with the paper \cite{E16}, for the reader convenience we also give the detailed of the proofs  of some important results.

\section{Multi-scale induction}

In the rest of the paper, we assume that  $n\geq 2$ and  $I_0$ is the interval from the previous section. 

Recall the following facts from \cite{E13}: Consider a cube $\BC^{(n)}_{L}(\Bu)$, with $\Bu=(u_1,\ldots,u_n)\in(\DZ^d)^n$.  We define
\[
\varPi\Bu=\{u_1,\ldots,u_n\},
\]
and 
\[
\varPi\BC^{(n)}_L(\Bu)=C^{(1)}_L(u_1)\cup\cdots\cup C^{(1)}_L(u_n).
\]
\begin{definition}
Let $L_0>3$ be a constant and $\alpha=3/2$. We define the sequence $\{L_k: k\geq 1\}$ recursively  as follows:
\[
L_k:=\lfloor L_{k-1}^{\alpha}\rfloor +1, \qquad \text{for all $k\geq 1$}.
\]
\end{definition}

Let $m>0$ a positive constant, we also introduce the following property, namely the multi-scale analysis bounds at any scale length $L_k$, and for any pair of separable cubes $\BC^{(n)}_{L_k}(\Bu)$ and $\BC^{(n)}_{L_k}(\Bv)$,
\begin{dsknn*}
\[
\prob{\exists E\in I_0: \text{$\BC^{(n)}_{L_k}(\Bu)$ and $\BC^{(n)}_{L_k}(\Bv)$ are $(E,m)$-S}}\leq L_k^{-2p4^{N-n}},
\]
where $p>6Nd$.
\end{dsknn*}

In both the single-particle and the multi-particle system, given the results on the multi-scale analysis property $\dsknn$ above one can deduce the localization results see for example the papers \cites{DK89,DS01} for those concerning the single-particle case and \cites{E13,CS09a} for multi-particle systems. We have the following

\begin{definition}[fully/partially interactive]\label{def:diagonal.cubes}
An $n$-particle cube $\BC_L^{(n)}(\Bu)\subset\DZ^{nd}$ is
called fully interactive (FI) if
\begin{equation}\label{eq:def.FI}
\diam \varPi \Bu := \max_{i\ne j} |u_i - u_j| \le n(2L+r_0),
\end{equation}
and partially interactive (PI) otherwise. 
\end{definition}

The following simple statement clarifies the notion of PI cube.

\begin{lemma}\label{lem:PI.cubes}
If a cube $\BC_L^{(n)}(\Bu)$ is PI, then there exists a subset $\CJ\subset\left\{1,\dots,n\right\}$ with $1\leq\card\CJ\leq n-1$ such that
\[
\dist\left(\varPi_{\CJ}\BC_L^{(n)}(\Bu),\varPi_{\CJ^{\comp}}\BC_L^{(n)}(\Bu)\right)>r_0,
\]
\end{lemma}  

\begin{proof}
See the appendix section \ref{sec:appendix}.
\end{proof}

If $\BC^{(n)}_L(\Bu)$ is a PI cube by the above Lemma, we can write it as 
\begin{equation}\label{eq:cartesian.cubes}
\BC_L^{(n)}(\Bu)=\BC_L^{(n')}(\Bu')\times\BC_L^{(n'')}(\Bu''),
\end{equation}
with 
\begin{equation}\label{eq:cartesian.cubes.2}
\dist\left(\varPi\BC_L^{(n')}(\Bu'),\varPi\BC_L^{(n'')}(\Bu'')\right)>r_0,
\end{equation}
where $\Bu'=\Bu_{\CJ}=(u_j:j\in\CJ)$, $\Bu''=\Bu_{\CJ^{\comp}}=(u_j:j\in\CJ^{\comp})$, $n'=\card \CJ$ and $n''=\card \CJ^{\comp}$. 

Throughout, when we write a PI cube $\BC^{(n)}_L(\Bu)$ in the form \eqref{eq:cartesian.cubes}, we implicitly assume that the projections satisfy \eqref{eq:cartesian.cubes.2}.
Let $\BC^{(n')}_{L_k}(\Bu')\times\BC^{(n'')}_{L_k}(\Bu'')$ be the decomposition of the PI cube $\BC^{(n)}_{L_k}(\Bu)$ and $\{\lambda_i,\varphi_i\}$ and $\{\mu_j,\phi_j\}$ be the eigenvalues and corresponding eigenfunctions of $\BH_{\BC_{L_k}^{(n')}(\Bu')}^{(n')}$ and $\BH_{\BC_{L_k}^{(n'')}(\Bu'')}^{(n'')}$ respectively. Next, we  can choose the eigenfunctions $\BPsi_{ij}$ of $\BH_{\BC^{(n)}_{L_k}(\Bu)}(\omega)$ as tensor products:
\[
\BPsi_{ij}=\varphi_i\otimes\phi_j
\]
The eigenfunctions appearing in subsequent arguments and calculation will be assumed normalized.

Now we turn to geometrical properties of FI cubes.
\begin{lemma}\label{lem:FI.cubes}
Let $n\geq 1$, $L>2r_0$ and consider two FI cubes $\BC_L^{(n)}(\Bx)$ and $\BC_L^{(n)}(\By)$ with $|\Bx-\By|>7\,nL$. Then
\begin{equation}
\varPi\BC_L^{(n)}(\Bx)\cap\varPi\BC_L^{(n)}(\By)=\varnothing.
\end{equation}
\end{lemma}

\begin{proof}
See the appendix section \ref{sec:appendix}.
\end{proof}

Given an $n$-particle cube $\BC_{L_{k+1}}^{(n)}(\Bu)$ and $E\in\DR$, we denote
\begin{itemize}
\item
by $M_{\pai}^{\sep}(\BC_{L_{k+1}}^{(n)}(\Bu),E)$ the maximal number of pairwise separable, 
$(E,m)$-singular PI cubes $\BC_{L_k}^{(n)}(\Bu^{(j)})\subset\BC_{L_{k+1}}^{(n)}(\Bu)$;
\item
by  $M_{\pai}(\BC_{L_{k+1}}^{(n)}(\Bu),E)$ the maximal number of (not necessarily separable)
$(E,m)$-singular PI cubes $\BC^{(n)}_{L_k}(\Bu^{(j)})$ contain in $\BC^{(n)}_{L_{k+1}}(\Bu)$ with $\Bu^{(j)}, \Bu^{(j')}\in\DZ^{nd}$ and $|\Bu^{(j)}-\Bu^{(j')}|>7NL_k$ for all $j\neq j'$;
\item
by $M_{\fui}(\BC_{L_{k+1}}^{(n)}(\Bu),E)$ the maximal number of 
$(E,m)$-singular FI cubes  $\BC_{L_k}^{(n)}(\Bu^{(j)})\subset \BC_{L_{k+1}}^{(n)}(\Bu)$ with $|\Bu^{(j)}-\Bu^{(j')}|>7NL_k$ for all $j\neq j'$\footnote{Note that by lemma \ref{lem:FI.cubes}, two FI cubes $\BC^{(n)}_{L_k}(\Bu^{(j)})$ and $\BC^{(n)}_{L_k}(\Bu^{(j')})$ with $|\Bu^{(j)}-\Bu^{(j')}|>7NL_k $ are automatically separable.},
\item
$M_{\pai}(\BC_{L_{k+1}}^{(n)}(\Bu),I):=\sup_{E\in I}M_{\pai}(\BC_{L_{k+1}}^{(n)}(\Bu),E)$.
\item
$M_{\fui}(\BC_{L_{k+1}}^{(n)}(\Bu),I):=\sup_{E\in I}M_{\fui}(\BC_{L_{k+1}}^{(n)}(\Bu),E)$.
\item
by $M(\BC_{L_{k+1}}^{(n)}(\Bu),E)$ the maximal number of 
$(E,m)$-singular cubes  $\BC_{L_k}^{(n)}(\Bu^{(j)})\subset \BC_{L_{k+1}}^{(n)}(\Bu)$ with $\dist(\Bu^{(j)},\partial\BC^{(n)}_{L_{k+1}}(\Bu))\geq 2L_k$ and $|\Bu^{(j)}-\Bu^{(j')}|>7NL_k$ for all $j\neq j'$.
\item
by $M^{\sep}(\BC_{L_{k+1}}^{(n)}(\Bu),E)$ the maximal number of pairwise separable
$(E,m)$-singular cubes  $\BC_{L_k}^{(n)}(\Bu^{(j)})\subset \BC_{L_{k+1}}^{(n)}(\Bu)$ 
\end{itemize}
Clearly
\[
M_{\pai}(\BC_{L_{k+1}}^{(n)}(\Bu),E)+M_{\fui}(\BC_{L_{k+1}}^{(n)}(\Bu),E)\geq M(\BC_{L_{k+1}}^{(n)}(\Bu),E).
\]

\subsection{Pairs of partially interactive cubes}\label{ssec:PI.cubes}
Let  $\BC_{L_{k+1}}^{(n)}(\Bu)=\BC^{(n')}_{L_{k+1}}(\Bu')\times\BC^{(n'')}_{L_{k+1}}(\Bu'')$ be a PI-cube. We also write $\Bx=(\Bx',\Bx'')$ for any point $\Bx\in\BC_{L_{k+1}}^{(n)}(\Bu)$, in the same way as $(\Bu',\Bu'')$. So  the corresponding Hamiltonian $\BH^{(n)}_{\BC_{L_{k+1}}^{(n)}(\Bu)}$ is written in the form:
\begin{equation}\label{eq:decomp,H}
\BH_{\BC_{L_{k+1}}^{(n)}(\Bu)}^{(n)}\BPsi(\Bx)=(-\BDelta\BPsi)(\Bx)+\left[\BU(\Bx')+\BV(\Bx',\omega)+\BU(\Bx'')+\BV(\Bx'',\omega)\right]\BPsi(\Bx)
\end{equation}
or, in compact form
\[
\BH_{\BC_{L_{k+1}}^{(n)}(\Bu)}^{(n)}=\BH_{\BC_{L_{k+1}}^{(n')}(\Bu')}^{(n')}\otimes\mathbf{I}+ \mathbf{I}\otimes \BH_{\BC_{L_{k+1}}^{(n'')}(\Bu'')}^{(n'')}.
\]
We denote by $\BG^{(n')}(\Bu',\Bv';E)$ and $\BG^{(n'')}(\Bu'',\Bv'';E)$ the corresponding Green functions, respectively. Introduce the following notions
\begin{definition}[\cite{KN13}]\label{def:HNR}
Let $1\leq n\leq N$ and $E\in\DR$. Consider a PI cube  $\BC^{(n)}_L(\Bu)=\BC^{(n')}_L(\Bu')\times\BC^{(n'')}_L(\Bu'')$. Then $\BC^{(n)}_L(\Bu)$ is called $E$-highly non resonant ($E$-HNR) if
\begin{enumerate}[\rm(i)]
\item 
for all $\mu_j\in\sigma(\BH^{(n'')}_{\BC^{(n'')}_L(\Bu'')})$, the cube  $\BC^{(n')}_L(\Bu')$ is $(E-\mu_j)$-CNR and 
\item
for all $\lambda_i\in\sigma(\BH^{(n')}_{\BH^{(n')}_L(\Bu')})$ the cube $\BC^{(n'')}_L(\Bu'')$ is $(E-\lambda_i)$-CNR.
\end{enumerate}
\end{definition}

\begin{definition}[$(E,m)$-tunnelling]    \label{def:tunnelling}
Let $1\leq n\leq N$,  $E\in\DR$ and $m>0$. Consider a PI cube $\BC^{(n)}_L(\Bu)=\BC^{(n')}_L(\Bu')\times \BC^{(n'')}_L(\Bu'')$.

Then $\BC^{(n)}_L(\Bu)$ is called
\begin{enumerate}[(i)]
\item
$(E,m)$ left-tunnelling ($(E,m)$-LT) if
$\exists \mu_j\in\sigma(\BH^{(n'')}_{\BC^{(n'')}_L(\Bu'')})$ such that $\BC^{(n')}_L(\Bu')$ contains two separable  $(E-\mu_j,m)$-S cubes $\BC^{(n')}_{l}(\Bv_1)$ and $\BC^{(n')}_{l}(\Bv_2)$ with $L=\lfloor l^{\alpha}\rfloor+1$. Otherwise, it is called $(E,m)$ non-left-tunnelling ($(E,m)$-NLT).
\item
$(E,m)$ right-tunnelling ($(E,m)$-RT) if  $\exists \lambda_i\in\sigma(\BH^{(n')}_{\BC^{(n')}_L(\Bu')})$ such that $\BC^{(n'')}_L(\Bu'')$ contains two separable $(E-\lambda_i,m)$-S cubes  $\BC^{(n'')}_{l}(\Bv_1)$ and $\BC^{(n'')}_{l}(\Bv_2)$ with $L=\lfloor l^{\alpha}\rfloor+1$. Otherwise, it is called $(E,m)$ non-right-tunnelling ($(E,m)$-NRT).
\item
$(E,m)$-tunnelling ($(E,m)$-T) if either it is $(E,m)$-LT or $(E,m)$-RT.
Otherwise it is called $(E,m)$-non-tunnelling ($(E,m)$-NT).
\end{enumerate}
\end{definition}

We reformulate and prove Lemma 3.18 from \cite{KN13} in our context.

\begin{lemma}\label{lem:HNR}
Let $E\in\DR$. If a PI cube $\BC^{(n)}_L(\Bu)=\BC^{(n')}_L(\Bu')\times\BC^{(n'')}_L(\Bu'')$  is not $E$-HNR then
\begin{enumerate}[\rm(i)]
\item
either there exist $L^{1/\alpha}\leq \ell\leq L$, $\Bx\in\BC^{(n')}_L(\Bu')$ such that  the $n$-particle rectangle $\BC^{(n)}=\BC^{(n')}_{\ell}(\Bx)\times\BC^{(n'')}_L(\Bu'')\subset\BC^{(n)}_L(\Bu)$ is $E$-R.
\item
or there exist $L^{1/\alpha}\leq \ell\leq L$, $\Bx\in\BC^{(n'')}_L(\Bu'')$ such that  the $n$-particle rectangle $\BC^{(n)}=\BC^{(n')}_{L}(\Bu')\times\BC^{(n'')}_{\ell}(\Bx)\subset\BC^{(n)}_L(\Bu)$ is $E$-R.
\end{enumerate}
\end{lemma}
\begin{proof}
By Definition \ref{def:HNR}, if $\BC^{(n)}_L(\Bu)$ is not $E$-HNR then either (a) there exists $\mu_j\in\sigma(\BH^{(n'')}_{\BC^{(n'')}_L(\Bu'')})$ such that  $\BC^{(n')}_L(\Bu')$ is not $E-\mu_j$-CNR or (b) there exists $\lambda_i\in\sigma(\BH^{(n')}_{\BC^{(n')}_L(\Bu')})$ such that $\BC^{(n'')}_L(\Bu'')$ is not $E-\lambda_i$-CNR. Let us first focus on case (a). Since $\BC^{(n')}_L(\Bu')$ is not $E-\mu_j$-CNR there exists $L^{1/\alpha}\leq \ell\leq L$, $\Bx\in \BC^{(n')}_L(\Bu')$ such that $\BC^{(n')}_{\ell}(\Bx)\subset\BC^{(n')}_L(\Bu')$ and $\BC^{(n')}_{\ell}(\Bx)$ is $E-\mu_j$-R. So $\dist(E-\mu_j,\sigma(\BH^{(n')}_{\BC^{(n')}_{\ell}(\Bx)}))<\ee^{-\ell^{\beta}}$. Therefore there exists $\eta \in\sigma(\BH^{(n')}_{\BC^{(n')}_{\ell}(\Bx)})$ such that $|E-\mu_j-\eta|<\ee^{-\ell^{\beta}}$. Now consider $\BC^{(n)}=\BC^{(n')}_{\ell}(\Bx)\times\BC^{(n'')}_L(\Bu'')$, since the cube $\BC^{(n)}_L(\Bu)$ is PI we have $\sigma(\BH^{(n)}_{\BC^{(n)}})=\sigma(\BH^{(n')}_{\BC^{(n')}_{\ell}(\Bx)})+\sigma(\BH^{(n'')}_{\BC^{(n'')}_L(\Bu'')})$, hence 
\[
\dist(E,\sigma(\BH^{(n)}_{\BC^{(n)}}))\leq |E-\mu_j-\eta|<\ee^{-\ell^{\beta}}.
\]
Thus $\BC^{(n)}$ is $E$-R. The same arguments shows that case (ii) arises when (b) occurs.
\end{proof}

\begin{lemma} \label{lem:NDRoNS}
Let $E\in I$ and $\BC_{L_k}^{(n)}(\Bu)$ be a PI cube. Assume that
$\BC_{L_k}^{(n)}(\Bu)$ is $(E,m)$-NT and $E$-HNR.
Then $\BC^{(n)}_{L_k}(\Bu)$ is $(E,m)$-NS.
\end{lemma}

\begin{proof}
Let $\BC^{(n')}_{L_k}(\Bu')\times\BC^{(n'')}_{L_k}(\Bu'')$ be the decomposition of the PI cube $\BC^{(n)}_{L_k}(\Bu)$.
Let $\{\lambda_i,\varphi_i\}$ and $\{\mu_j,\phi_j\}$ be the eigenvalues and corresponding eigenvectors of $\BH_{\BC_{L_k}^{(n')}(\Bu')}^{(n')}$ and $\BH_{\BC_{L_k}^{(n'')}(\Bu'')}^{(n'')}$ respectively. Then we can choose the eigenvectors $\BPsi_{ij}$ and corresponding eigenvalues $E_{ij}$ of $\BH_{\BC^{(n)}_{L_k}(\Bu)}(\omega)$ as follows.
\[
\BPsi_{ij}=\varphi_i\otimes\phi_j,\qquad E_{ij}=\lambda_i+\mu_j.
\]
By the assumed  $E$-HNR property of the cube $\BC^{(n)}_{L_k}(\Bu)$, for all eigenvalues $\lambda_i$ one has $\BC^{(n'')}_L(\Bu'')$ is  $E-\lambda_i$-CNR. Next,  by assumption of $(E,m)$-NT, $\BC^{(n'')}_{L_k}(\Bu'')$ does not contain any pair of separable $(E-\lambda_i,m)$-S cubes of radius $L_{k-1}$ therefore by Lemma \ref{lem:MPI} $M(\BC^{(n)}_{L_{k+1}}(\Bu),E-\lambda_i)<\kappa(n)+2$ and Lemma \ref{lem:CNR.NS} implies that it is also $(E-\lambda_i,m)$-NS, yielding 
\begin{equation}\label{eq:lem.NDRoNS.1}
\max_{\{\lambda_i\}}
\max_{\Bv''\in \partial^- \BC_{L_k}^{(n'')}(\Bu'')}
\left|\BG^{(n'')}(\Bu'',\Bv'';E-\lambda_i)\right|\leq\ee^{-\gamma(m,L_k,n'')L_k}.
\end{equation}
The same analysis for $\BC^{(n')}_L(\Bu')$ also gives
\begin{equation}
\label{eq:lem.NDRoNS.2}
\max_{\{ \mu_j\}}
      \max_{\Bv'\in \partial^-\BC_{L_k}^{(n')}(\Bu')}
      \left|\BG^{(n')}(\Bu',\Bv';E-\mu_j)\right|\leq\ee^{-\gamma(m,L_k,n')L_k}.
\end{equation}
For any $\Bv\in\partial^-\BC^{(n)}_{L_k}(\Bu)$, $|\Bu-\Bv|=L_k$, thus either $|\Bv'-\Bu'|=L_k$ or $|\Bv''-\Bu''|=L_k$. Consider first the latter case. Equation \eqref{eq:lem.NDRoNS.1} applies and we get
\begin{gather*}
\left|\BG^{(n)}(\Bu,\Bv;E)\right| =\left|\sum_{i,j} \frac{\varphi_i(\Bu')\varphi_i(\Bv')\phi_j(\Bu'')\phi_j(\Bv'')}{E-\lambda_i-\mu_j}\right|\\
\leq \sum_i \left|\varphi_i(\Bu')\varphi_i(\Bv')\right|\cdot \left|\BG^{(n)}(\Bu'',\Bv'';E-\lambda_i)\right|\\
\leq (2L_k+1)^{(n-1)d}
\max_{\{ \lambda_i \}}\;
\max_{\Bv''\in\partial\BC^{(n)}_{L_k}(\Bu'') }
\left|\BG^{(n)}(\Bu'',\Bv'';E-\lambda_i)\right|, \tag{ since $\|\varphi\|_{\infty}\leq 1$}\\
\leq (2L_k+1)^{(n-1)d}\cdot\ee^{-\gamma(m,L_k,n-1)L_k}\\
=\ee^{-[\gamma(m,L_k,n-1)-L_k^{-1}\ln(2L_k+1)^{(n-1)d}]L_k}.
\end{gather*}
But by  Definition \eqref{eq:gamma}:
\[
\gamma(m,L_k,n)=m(1+L_k^{-1/8})^{N-n+1},
\]
For $2\leq n\leq N$,
\[
\gamma(m,L_k,n-1)-\gamma(m,L_k,n)>L_k^{-1}\ln(2L_k+1)^{(n-1)d}.
\]
Indeed, setting $C_1=\frac{2^{-N}\gamma}{3\sqrt{2}}$,
\begin{align*}
\gamma(m,L_k,n-1)-\gamma(m,L_k,n)&=mL_k^{-1/8}(1+L_k^{-1/8})^{N-n+1}\\
&=C_1L_0^{-1/2}L_k^{-1/8}(1+L_k^{-1/8})^{N-n+1}>C_1L_k^{-5/8},
\end{align*}
and for $L_0$  sufficiently large, hence $L_k$,
\[
L_k^{-1}\ln(2L_k+1)^{(n-1)d}\leq L_k^{-1}(n-1)d(3L_k)^{3/8}\leq C_1 L_k^{-5/8}.
\]
Thus, $\BC^{(n)}_{L_k}(\Bu)$ is $(E,m)$-NS. Finally, the case $|\Bu'-\Bv'|=L_k$ is similar.
\end{proof}

\begin{lemma}\label{lem:T.estimate}
Let $2\leq n\leq N$ and assume property $\dsn{k,n'}$ for any $1\leq n'<n$.
Then for any  PI cube $\BC_{L_{k+1}}^{(n)}(\By)$ one has
\begin{equation}\label{eq:C.is.T}
\DP\bigl\{\exists E\in I, \BC_{L_{k+1}}^{(n)}(\By)\text{ is $(E,m)$-T}\bigr\}\leq \frac{1}{2}L_{k+1}^{-4p\,4^{N-n}}.
\end{equation}
\end{lemma}

\begin{proof}
Consider a PI cube  $\BC_{L_{k+1}}^{(n)}(\By)=\BC_{L_{k+1}}^{(n')}(\By')\times\BC_{L_{k+1}}^{(n'')}(\By'')$.
By definition \ref{def:tunnelling}, we have that the event
\[
\left\{\exists E\in I:\BC_{L_{k+1}}^{(n)}(\By)\text{ is $(E,m)$-T}\right\},
\]
is contained in the union
\[
\left\{\exists E\in I: \BC^{(n)}_{L_{k+1}}(\By)\text{ is $(E,m)$-RT}\right\}\cup\left\{\exists E\in I: \BC^{(n)}_{L_{k+1}}(\By)\text{ is $(E,m)$-LT}\right\}.
\]
Now, since $E\in I$ and $\mu_j\geq 0$ we have $ E-\mu_j\leq E^*$. So for any $j$, $E-\mu_j\in I$.
Further using property $\dsn{k,n'}$ we have
\begin{align*}
\prob{\text{$\exists E\in I$, $\BC^{(n)}_{L_{k+1}}(\By)$ is $(E,m)$-RT}}
&\leq \frac{|\BC^{(n')}_{L_{k+1}}(\By')|^{2}}{2}|\BC^{(n'')}_{L_{k+1}}(\By'')|L_k^{-2p4^{N-n'}}\\
&\leq C(n,N,d)L_{k+1}^{-2p\,\frac{4^{N-(n-1)}}{\alpha}+3(n-1)d}.
\end{align*}
A similar argument also shows that
\[
\prob{\text{$\exists E\in I$, $\BC^{(n)}_{L_{k+1}}(\By)$ is $(E,m)$-LT}}
\leq C(n,N,d)L_{k+1}^{-2p\,\frac{4^{N-(n-1)}}{\alpha}+3(n-1)d},
\]
so that
\[
\prob{\exists E\in I: \BC^{(n)}_{L_{k+1}}(\By)\text{ is $(E,m)$-T}}
\leq C(n,N,d)L_{k+1}^{-2p\,\frac{4^{N-(n-1)}}{\alpha}+3(n-1)d}.
\]
The assertion follows by observing that $2p\,4^{N-(n-1)}/\alpha \, -3(n-1)d>4p\,4^{N-n}$ for $\alpha=3/2$ provided
$L_0$ is large enough and  $p>4\alpha Nd=6Nd$.
\end{proof}

\begin{theorem}\label{thm:partially.interactive}
Let $1\leq n\leq N$. There exists $L_1^*=L_1^*(N,d)>0$ such that if $L_0\geq L_1^*$ and if for $k\geq 0$ $\dsn{k,n'}$ holds true for any $1\leq n'<n$, then $\dskonn$ holds true
for any pair of separable PI cubes $\BC_{L_{k+1}}^{(n)}(\Bx)$ and $\BC_{L_{k+1}}^{(n)}(\By)$.
\end{theorem}

\begin{proof}
Let $\BC_{L_{k+1}}^{(n)}(\Bx)$ and $\BC_{L_{k+1}}^{(n)}(\By)$ be two separable PI-cubes. Consider the events:
\begin{align*}
\rB_{k+1}
&=\bigl\{\exists\,E\in I:\BC_{L_{k+1}}^{(n)}(\Bx)\text{ and $\BC_{L_{k+1}}^{(n)}(\By)$ are $(E,m)$-S}\bigr\},\\
\rR
&=\bigl\{\exists\,E\in I:\text{neither $\BC_{L_{k+1}}^{(n)}(\Bx)$ nor $\BC_{L_{k+1}}^{(n)}(\By)$ is $E$-HNR}\bigr\},\\
\rT_{\Bx}
&=\bigl\{\exists E\in I: \BC_{L_{k+1}}^{(n)}(\Bx)\text{ is $(E,m)$-T}\bigr\},\\
\rT_{\By}
&=\bigl\{\exists E\in I: \BC_{L_{k+1}}^{(n)}(\By)\text{ is $(E,m)$-T}\bigr\}.
\end{align*}
If $\omega\in\rB_{k+1}\setminus\rR$, then $\forall E\in I$, $\BC_{L_{k+1}}^{(n)}(\Bx)$ or $\BC_{L_{k+1}}^{(n)}(\By)$ is $E$-HNR. If $\BC_{L_{k+1}}^{(n)}(\By)$ is $E$-HNR, then it must be $(E,m)$-T: otherwise it would have been $(E,m)$-NS by Lemma~\ref{lem:NDRoNS}. Similarly, if $\BC_{L_{k+1}}^{(n)}(\Bx)$ is $E$-HNR, then it must be $(E,m)$-T. This implies that
\[
\rB_{k+1}\subset\rR\cup\rT_{\Bx}\cup\rT_{\By}.
\]
Therefore,
\begin{align*}
\DP\left\{\rB_{k+1}\right\}&\leq\DP\left\{\rR\right\}+\DP\{\rT_{\Bx}\}+\DP\{\rT_{\By}\}\\
&\leq \prob{\rR}+\frac{1}{2}L_{k+1}^{-4p\,4^{N-n}}+\frac{1}{2}L_{k+1}^{-4p\,4^{N-n}}\\
\end{align*}
where we used  \eqref{eq:C.is.T} to estimate $\DP\{\rT_{\Bx}\}$ and $\DP\{\rT_{\By}\}$. Next by combining Theorem \ref{thm:Wegner} and Lemma \ref{lem:HNR} we obtain that $\prob{\rR}\leq L_{k+1}^{-4^N\,p}$. Finally 
\begin{equation}\label{eq:bound.PI}
\prob{\rB_{k+1}}\leq L_{k+1}^{-4^Np}+L_{k+1}^{-4p4^{N-n}}<L_{k+1}^{-2p4^{N-n}}.
\end{equation}
\end{proof}

For subsequent calculations and proofs we give the following two Lemmas. 

\begin{lemma}\label{lem:MPI}
If $M(\BC^{(n)}_{L_{k+1}}(\Bu),E)\geq \kappa(n)+2$ with $\kappa(n)=n^n$, then $M^{\sep}(\BC^{(n)}_{L_{k+1}}(\Bu),E)\geq 2$.\\
\noindent Similarly, if $M_{\pai}(\BC^{(n)}_{L_{k+1}}(\Bu),E)\geq \kappa(n)+2$  then $M_{\pai}^{\sep}(\BC^{(n)}_{L_{k+1}}(\Bu),E)\geq 2$.
\end{lemma}
\begin{proof}
See the appendix section \ref{sec:appendix}.
\end{proof}

\begin{lemma}\label{lem:MND}
With the above notations, assume that $\dsn{k-1,n'}$ holds true for all $1\leq n'<n$ then
\begin{equation}\label{eq:MND}
\DP\left\{M_{\pai}(\BC_{L_{k+1}}^{(n)}(\Bu),I)\geq \kappa(n)+ 2\right\}\leq \frac{3^{2nd}}{2} L_{k+1}^{2nd}\left(L_k^{-4^Np}+L_k^{-4p\,4^{N-n}}\right).
\end{equation}
\end{lemma}	
\begin{proof}
See the appendix section \ref{sec:appendix}.
\end{proof}

\subsection{Pairs of fully interactive cubes}\label{ssec:FI.cubes}
Our aim now is to prove $\dskonn$ for a pair of separable fully interactive cubes $\BC_{L_{k+1}}^{(n)}(\Bx)$ and $\BC_{L_{k+1}}^{(n)}(\By)$. We adapt to the continuum a very crucial and hard result obtained in the paper \cite{E13} and which generalized to multi-particle systems some previous work by von Dreifus and Klein \cite{DK89} on the lattice and Stollmann \cite{St01} in the continuum for single-particle models. 

\begin{lemma}\label{lem:CNR.NS}
Let $J=\kappa(n)+5$ with $\kappa(n)=n^n$ and $E\in\DR$. Suppose that
\begin{enumerate}[\rm(i)]
\item
$\BC_{L_{k+1}}^{(n)}(\Bx)$ is $E$-CNR,
\item
$M(\BC_{L_{k+1}}^{(n)}(\Bx),E)\leq J$.
\end{enumerate}
Then there exists $\tilde{L}_2^*(J,N,d)>0$ such that if $L_0\geq \tilde{L}_{2}^*(J,N,d)$ we have that $\BC^{(n)}_{L_{k+1}}(\Bx)$ is $(E,m)$-NS.
\end{lemma}
\begin{proof}
Since, $M(\BC^{(n)}_{L_{k+1}}(\Bx);E)\leq J$, there exists at most $J$ cubes of side length $2L_k$ contained in $\BC^{(n)}_{L_{k+1}}(\Bx)$ that are $(E,m)$-S with centers at distance $> 7NL_k$. Therefore, we can find $\Bx_i\in\BC^{(n)}_{L_{k+1}}(\Bu)\cap\Gamma_{\Bx}$ with $\Gamma_{\Bx}=\Bx+\frac{L_k}{3}\DZ^{nd}$
\[
\dist(\Bx_i,\partial\BC^{(n)}_{L_{k+1}}(\Bx))\geq 2L_k,\quad\text{ $i=1,\ldots,r\leq J$},
\]
such that, if $\Bx_0\in \BC^{(n)}_{L_{k+1}}(\Bx)\setminus\bigcup_{i=1}^r\BC^{(n)}_{2L_k}(\Bx_i)$, then the cube $\BC^{(n)}_{L_k}(\Bx_0)$ is $(E,m)$-NS.

We do an induction procedure in $\BC^{(n,int)}_{L_{k+1}}(\Bx)$ and start with $\Bx_0\in\BC^{(n,int)}_{L_{k+1}}(\Bx)$. We estimate $\|\Bone_{\BC^{(n,out)}_{L_{k+1}}(\Bx)}\BG^{(n)}_{L_{k+1}}(E)\Bone_{\BC^{(n,int)}_{L_k}(\Bx_0)}\|$. Suppose that $\Bx_0,\ldots,\Bx_{\ell}$ have been choosen for $\ell\geq 0$. We have two cases:

\begin{enumerate}
\item[case(a)] $\BC^{(n)}_{L_k}(\Bx_{\ell})$ is $(E,m)$-NS.\\
 In this case, we apply the (GRI) Theorem \ref{thm:GRI.GF} and obtain
\begin{gather*}
\|\Bone_{\BC^{(n,out)}_{L_{k+1}}(\Bx)}\BG^{(n)}_{\BC^{(n)}_{L_{k+1}}(\Bx)}(E)\Bone_{\BC^{(n,int)}_{L_{k}}(\Bx_0)}\|\\
\leq C_{geom}\|\Bone_{\BC^{(n,out)}_{L_{k+1}}(\Bx)}\BG^{(n)}_{\BC^{(n)}_{L_{k+1}}(\Bx)}(E)\Bone_{\BC^{(n,out)}_{L_k}(\Bx_0)}\|\cdot\|\Bone_{\BC^{(n,out)}_{L_k}(\Bx)}\BG^{(n)}_{\BC^{(n)}_{L_k}(\Bx_0)}(E)\Bone_{\BC^{(n,int)}_{L_k}(\Bx_0)}\|\\
\leq C_{geom}\|\Bone_{\BC^{(n,out)}_{L_{k+1}}(\Bx)}\BG^{(n)}_{\BC^{(n)}_{L_{k+1}}}(E)\Bone_{\BC^{(n,out)}_{L_k}(\Bx)}\|\cdot\ee^{-\gamma(m,L_k,n)L_k}.
\end{gather*}
 We replace in the above analysis $\Bx$ with $\Bx_{\ell}$ and we get 
\[
\|\Bone_{\BC^{(n,out)}_{L_{k+1}}(\Bx_{\ell})}\BG^{(n)}_{\BC^{(n)}_{L_{k+1}}(\Bx_{\ell})}(E)\Bone_{\BC^{(n,out)}_{L_k}(\Bx_{\ell})}\|\leq 3^{nd}\|\Bone_{\BC^{(n,out)}_{L_{k+1}}(\Bx_{\ell})}\BG^{(n)}_{\BC^{(n)}_{L_{k+1}}\Bx_{\ell}}(E)\Bone_{\BC^{(n,int)}_{L_{k}}(\Bx_{\ell+1}}\|,
\]
where $\Bx_{\ell+1}$ is choosen in such a way that the norm in the right hand side in the above equation is maximal. Observe that $|\Bx_{\ell}-\Bx_{\ell+1}|=L_k/3$. We therefore obtain

\begin{gather*}
\|\Bone_{\BC^{(n,out)}_{L_{k+1}}(\Bx)}\BG^{(n)}_{\BC^{(n)}_{L_{k+1}}}(E)\Bone_{\BC^{(n,int)}_{L_k}(\Bx_{\ell}}\|\\
\leq C_{geom}3^{nd}\ee^{-\gamma(m,L_k,n)L_k}\cdot\|\Bone_{\BC^{(n,out)}_{L_{k+1}}(\Bx)}\BG^{(n)}_{\BC^{(n)}_{L_{k+1}}\Bx)}(E)\Bone_{\BC^{(n,int)}_{L_k}(\Bx_{\ell+1})}\|\\
\leq \delta_{+}\|\Bone_{\BC^{n,out)}_{L_{k+1}}(\Bx)}\BG^{(n)}_{\BC^{(n)}_{L_{k+1}}(\Bx)}(E)\Bone_{\BC^{(n,int)}_{L_k}(\Bx_{\ell+1})}\|
\end{gather*}
with 
\[
\delta_{+}=3^{nd}C_{geom}\ee^{-\gamma(m,L_k,n)L_k}.
\]

\item[case(b)] $\BC^{(n)}_{L_k}(\Bx_{\ell})$ is  $(E,m)$-S. \\
Thus, there exists $i_0=1,\ldots,r$ such that $\BC^{(n)}_{L_k}(\Bx_{\ell})\subset\BC^{(n)}_{2L_k}(\Bx_{i_0})$. We apply again the (GRI) this time with $\BC^{n)}_{L_{k+1}}(\Bx)$ and $\BC^{(n)}_{2L_k}(\Bx_{i_0})$ and obtain
\begin{gather*}
\|
\Bone_{\BC^{(n,out)}_{L_{k+1}}(\Bx)}\BG^{(n)}_{\BC^{(n)}_{L_{k+1}}(\Bx)}(E)\Bone_{\BC^{(n,int)}_{2L_k}(\Bx_{i_0})}\|\leq C_{geom}\|\Bone_{\BC^{(n,out)}_{L_{k+1}}(\Bx)}\BG^{(n)}_{\BC^{(n)}_{L_{k+1}}(\Bx)}(E)\Bone_{\BC^{(n,out)}_{L_k}(\Bx_{i_0})}\|\\
\times  \|\Bone_{\BC^{(n,out)}_{L_k}(\Bx_{i_0})}\BG^{(n)}_{\BC^{(n)}_{L_k}(\Bx_{i_0})}(E)\Bone_{\BC^{(n,int)}_{L_k}(\Bx_{i_0})}\|\\
\leq C_{geom} \ee^{(2L_k)^{1/2}}\cdot \|\Bone_{\BC^{(n,out)}_{L_{k+1}}(\Bx)}\BG^{(n)}_{\BC^{(n)}_{L_{k+1}}(\Bx)}(E)\Bone_{\BC^{(n,out)}_{2L_k}(\Bx_{i_0})}\|
\end{gather*}

We have almost everywhere
\[
\Bone_{\BC^{(n,out)}_{2L_k}(\Bx_{i_0})}\leq \sum_{\tilde{\Bx}\in\BC^{(n)}_{2L_k}(\Bx_{i_0})\cap\Gamma_{\Bx_{i_0}},\BC^{(n)}_{L_k}(\tilde{\Bx})\not\subset\BC^{(n)}_{2L_k}(\Bx_{i_0})}\Bone_{\BC^{(n,int)}_{L_k}(\tilde{\Bx})}
\]
Hence, by choosing $\tilde{\Bx}$ such that the right hand side is maximal, we get
\[
\|\Bone_{\BC^{(n,out)}_{L_{k+1}}(\Bx)}\BG^{(n)}_{\BC^{(n)}_{L_{k+1}}(\Bx)}(E)\Bone_{\BC^{(n,int)}_{2L_k}(\Bx_{i_0})}\|\leq 6^{nd}\cdot\|\Bone_{\BC^{(n,out)}_{L_{k+1}}(\Bx)}\BG^{(n)}_{\BC^{(n)}_{L_{k+1}}(\Bx)}(E)\Bone_{\BC^{(n,int)}_{L_k}(\tilde{\Bx})}\|.
\]
Since, $\BC^{(n)}_{L_k}(\tilde{\Bx})\not\subset\BC^{(n)}_{2L_k}(\Bx_{i_0})$ , $\tilde{\Bx}\in\BC^{(n)}_{2L_k}(\Bx_{i_0})$ and the cubes $\BC^{(n)}_{2L_k}(\Bx_i)$ are disjoint, we obtain that
\[
\BC^{(n)}_{L_k}(\tilde{\Bx})\not\subset \bigcup_{i=1}^r\BC^{(n)}_{2L_k}(\Bx_i),
\]
so that the cube $\BC^{(n)}_{L_k}(\tilde{\Bx})$ must be $(E,m)$-NS. We therefore perform a new step as in case (a) and obtain:

\[
\cdots \leq 6^{nd}3^{nd}C_{geom}\cdot \ee^{-\gamma(m,L_k,n)L_k}\cdot\|\Bone_{\BC^{(n,out)}_{L_{k+1}}(\Bx)}\BG^{(n)}_{\BC^{(n)}_{L_{k+1}}(\Bx)}(E)\Bone_{\BC^{(n,int)}_{L_k}(\Bx_{\ell+1})}\|,
\]

with $\Bx_{\ell+1}\in\Gamma_{\tilde{\Bx}}$ and $|\tilde{\Bx}-\Bx_{\ell+1}|=L_k/3$.
\end{enumerate}

 Summarizing, we get $\Bx_{\ell+1}$ with 

\[
\|\Bone_{\BC^{(n,out)}_{L_{k+1}}(\Bx)}\BG^{(n)}_{\BC^{(n)}_{L_{k+1}}(\Bx)}(E)\Bone_{\BC^{(n,int)}_{L_k}(\Bx_{\ell})}\|\leq \delta_{0}\cdot\|\Bone_{\BC^{(n,out)}_{L_{k+1}}(\Bx)}\BG^{(n)}_{\BC^{(n)}_{L_{k+1}}(\Bx)}(E)\Bone_{\BC^{(n,int)}_{L_{k+1}}(\Bx_{\ell+1})}\|,
\]
with $\delta_0=18^{nd} C_{geom}^2\cdot\ee^{(2L_k)^1/2}\ee^{-\gamma(m,L_k,n)L_k}$. After $\ell$ iterations  with $n_+$ steps of case (a) and $n_0$ steps of case (b), we obtain

\[
\|\Bone_{\BC^{(n,out)}_{L_{k+1}}(\Bx)}\BG^{(n)}_{\BC^{(n)}_{L_{k+1}}(\Bx)}(E)\Bone_{\BC^{(n,int)}_{L_k}(\Bx_0)}\|\leq (\delta_+)^{n_+}(\delta_0)^{n_0}\cdot\|\Bone_{\BC^{(n,out)}_{L_{k+1}}(\Bx)}\BG^{(n)}_{\BC^{(n)}_{L_{k+1}}(\Bx)}(E)\Bone_{\BC^{(n,int)}_{L_k}(\Bx_{\ell})}\|.
\]
Now since $\gamma(m,L_k,n)>m$, we have that 
\[
\delta_+\leq 3^{nd}\cdot C_{geom}\ee^{-mL_k}.
\]
So $\delta_+$ can be made arbitrarily  small if $L_0$ and hence $L_k$ is large enough. We also have for $\delta_0$:
\begin{align*}
\delta_0&=18^{nd} C_{geom}^2\ee^{(2L_k)^1/2}\ee^{-\gamma(m,n,L_k)L_k}\\
&=18^{nd} C_{geom}^2\ee^{\sqrt{2}L_k^{1/2}}\ee^{-\gamma(m,n,L_k)L_k}\\
&\leq 18^{nd} C_{geom}^2\ee^{\sqrt{2}L_k^{1/2}-mL_k}<\frac{1}{2},
\end{align*}
For large $L_0$ and hence $L_k$. Using the (GRI), we can iterate if $\BC^{(n,out)}_{L_{k+1}}(\Bx)\cap\BC^{(n)}_{L_k}(\Bx_{\ell})=\emptyset$. Thus, we can have at least $n_+$ steps of case (a) with,
\[
n_+\cdot\frac{L_k}{3}+\sum_{i=1}^r2L_k\geq \frac{L_{k+1}}{3}-\frac{L_k}{3},
\]
until the induction eventually stop. Since $r\leq J$, we can bound $n_{+}$ from below .
\begin{align*}
n_+\cdot\frac{L_k}{3}&\geq \frac{L_{k+1}}{3}-\frac{L_k}{3}-r(L_k)\\
&\geq \frac{L_{k+1}}{3}-\frac{L_k}{3}-2JL_k\\
\end{align*}
which yields
\begin{align*}
n_+&\geq \frac{L_{k+1}}{L_k}-1-6J\\
&\geq \frac{L_{k+1}}{L_k}-7J
\end{align*}
Therefore,
\begin{equation}\label{eq:NR.NS}
\|\Bone_{BC^{(n,out)}_{L_{k+1}}(\Bx)}\BG^{(n)}_{\BC^{(n)}_{L_{k+1}}(\Bx)}(E)\Bone_{\BC^{(n,int)}_{L_k}(\Bx_0)}\|\leq \delta_+^{n_+}\cdot\|\BG^{(n)}_{\BC^{(n)}_{L_{k+1}}(\Bx)}(E)\|.
\end{equation}
Finally, by $E$-nonresonance of $\BC^{(n)}_{L_{k+1}}(\Bx)$ and since we can cover  $\BC^{(n,int)}_{L_{k+1}}(\Bx)$ by $\left(\frac{L_{k+1}}{L_k}\right)^{nd}$ small cubes $\BC^{(n,int)}_{L_k}(\By)$, equation \eqref{eq:NR.NS} with $y$ instead of $\Bx_0$ yields

\begin{gather*}
\|\Bone_{\BC^{(n,out)}_{L_{k+1}}(\Bx)}\BG^{(n)}_{\BC^{(n)}_{L_{k+1}}(\Bx)}(E)\Bone_{\BC^{(n,int)}_{L_{k+1}}(\Bx)}\| \\
\leq \left(\frac{L_{k+1}}{L_k}\right)^{nd}\cdot\delta_{n_+}\cdot\ee^{L_{k+1}^{1/2}}\\
\leq \left(\frac{L_{k+1}}{L_k}\right)^{nd}\cdot\left[3^{nd}\cdot Cgeom\cdot\ee^{-\gamma(m,L_k,n)L_k}\right]^{\frac{L_{k+1}}{L_k}-7J}\ee^{L_{k+1}^{1/2}}\\
\leq L_{k+1}^{nd}L_{k+1}^{-\frac{nd}{\alpha}} C(n,d)^{\frac{L_{k+1}}{L_k}-7J}\ee^{-\gamma(m,L_k,n)(\frac{L_{k+1}}{L_k}-7J)}\times\ee^{L_{k+1}^{1/2}}\\
\leq L_{k+1}^{nd/3}\ee^{(L_{k+1}^{1/3}-7J)\ln C(n,d)}\ee^{-\gamma(m,L_k,n)(L_{k+1}^{1/3}-7J)}\ee^{L_{k+1}^{1/2}}\\
\leq \ee^{-\left[-\frac{nd}{3}\ln(L_{k+1})-L_{k+1}^{1/3}\ln(C)+7J\ln(C(n,d))+7J\ln(C(n,d))+\gamma(m,L_k,n)L_{k+1}^{1/3}-7J\gamma(m,L_k,n)-L_{k+1}^{1/2}\right]}\\
\leq \ee^{-\left[\frac{-nd}{3}\frac{\ln L_{k+1}}{L_{k+1}}-\frac{L_{k+1}^{1/3}\ln C(n,d)}+\frac{7J\ln(C(n,d))}{L_{k+1}}+\gamma(m,L_k,n)\frac{L_{k+1}^{1/3}}{L_{k+1}}-7J\frac{\gamma(m,L_k,n)}{L_{k+1}}-{L_{k+1}}^{-1/2}\right]L_{k+1}}\\
\leq \ee^{-m'L_{k+1}},
\end{gather*}
where 
\[
m'=\frac{1}{L_{k+1}}\left[ n_+\gamma(m,L_k,n)L_k-n_+\ln((2^{Nd}NdL_k^{nd-1})\right]-\frac{1}{L_{k+1}^{1/2}},
\]
with
\[
L_{k+1}L_k^{-1}-7J\leq n_+\leq L_{k+1}L_{k}^{-1}; 
\]
we obtain
\begin{align*}
m'&\geq \gamma(m,L_k,n)-\gamma(m,L_k,n)\, \frac{4JL_k}{L_{k+1}}
\\
\qquad &-\frac{1}{L_{k+1}}\frac{L_{k+1}}{L_k}\ln(2^{Nd}Nd)L_k^{nd-1})-\frac{1}{L_{k+1}^{1/2}}
\\
&\geq \gamma(m,L_k,n)-\gamma(m,L_k,n)\, 4J L_k^{-1/2}
\\
&\qquad -L_k^{-1}(\ln(2^{Nd}Nd))+(nd-1)\ln(L_k))-L_{k}^{-3/4}
\\
&\geq \gamma(m,L_k,n)[1-(4J+\ln(2^{Nd}Nd)+Nd)L_k^{-1/2}]
\end{align*}

if $L_0\geq L_2^*(J,N,d)$ for some $L^*_2(J,N,d)>0$ large enough. Since $\gamma(m,L_k,n)=m(1+L_k^{-1/8})^{N-n+1}$,
\[
\frac{\gamma(m,L_k,n)}{\gamma(m,L_{k+1},n)}=\left(\frac{1+L_k^{-1/8}}{1+L_k^{-3/16}}\right)^{N-n+1}\geq \frac{1+L_k^{-1/8}}{1+L_k^{-3/16}}
\]
Therefore we can compute
\begin{align*}
&\frac{\gamma(m,L_k,n)}{\gamma(m,L_{k+1},n)}(1-(4J+\ln(2^{Nd}Nd)+Nd)L_k^{-1/2})\\
 &\qquad \geq\frac{1+L_k^{-1/8}}{1+L_k^{-3/16}}(1-(4J+\ln(2^{Nd}Nd)+Nd)L_k^{-1/2})>1,
\end{align*}
provided $L_0\geq \tilde{L}_2^*(J,N,d)$ for some large enough $\tilde{L}_2^*(J,N,d)>0$. Finally, we obtain that $m'>\gamma(m,L_{k+1},n)$ and $|\BG_{\BC^{(n)}_{L_{k+1}}(\Bu)}(\Bu,\Bv;E)|\leq \ee^{-\gamma(m,L_{k+1},n)L_{k+1}}$. This proves the result. 
\end{proof}

The main result of this subsection is Theorem \ref{thm:fully.interactive}. We will  need the following preliminary results.

\begin{lemma}\label{lem:MD}
Given $k\geq0$, assume that property $\dsknn$ holds true for all pairs of separable FI cubes. Then for any $\ell\geq 1$
\begin{equation}\label{eq:MD}
\DP\left\{M_{\fui}(\BC_{L_{k+1}}^{(n)}(\Bu),I)\geq 2\ell\right\}\leq C(n,N,d,\ell)L_k^{2\ell dn\alpha}L_k^{-2\ell p\,4^{N-n}}.
\end{equation}
\end{lemma}
\begin{proof}
See the proof in the appendix Section \ref{sec:appendix}.
\end{proof}

\begin{theorem}\label{thm:fully.interactive}
Let $1\leq n\leq N$. There exists $L_2^*=L_2^*(N,d)>0$ such that if $L_0\geq L_2^*$ and if for $k\geq 0$
\begin{enumerate}[\rm(i)]
\item
$\dsn{k-1,n'}$ for all $1\leq n'<n$ holds true,
\item
$\dsn{k,n}$ holds true for all pairs of FI cubes,
\end{enumerate}
then $\dskonn$ holds true
for any pair of separable FI cubes $\BC_{L_{k+1}}^{(n)}(\Bx)$ and $\BC_{L_{k+1}}^{(n)}(\By)$.
\end{theorem}

Above, we used the convention that $\dsn{-1,n}$ means no assumption.

\begin{proof}
Consider a pair of separable FI cubes $\BC_{L_{k+1}}^{(n)}(\Bx)$, $\BC_{L_{k+1}}^{(n)}(\By)$ and set $J=\kappa(n)+5$. Define
\begin{align*}
\rB_{k+1}
&=\left\{\exists\,E\in I_0:\BC_{L_{k+1}}^{(n)}(\Bx)\text{ and $\BC_{L_{k+1}}^{(n)}(\By)$ are $(E,m)$-S}\right\},\\
\Sigma
&=\left\{\exists\,E\in I_0:\text{neither $\BC_{L_{k+1}}^{(n)}(\Bx)$ nor $\BC_{L_{k+1}}^{(n)}(\By)$ is $E$-CNR}\right\},\\
\rS_{\Bx}
&=\left\{\exists\,E\in I_0:M(\BC_{L_{k+1}}^{(n)}(\Bx);E)\geq J+1\right\},\\
\rS_{\By}
&=\left\{\exists\,E\in I_0:M(\BC_{L_{k+1}}^{(n)}(\By),E)\geq J+1\right\}.
\end{align*}
Let $\omega\in\rB_{k+1}$. If $\omega\notin\Sigma\cup\rS_{\Bx}$, then $\forall E\in I_0$ either $\BC_{L_{k+1}}^{(n)}(\Bx)$ or $\BC_{L_{k+1}}^{(n)}(\By)$ is $E$-CNR and $M(\BC_{L_{k+1}}^{(n)}(\Bx),E)\leq J$. The cube $\BC_{L_{k+1}}^{(n)}(\Bx)$ cannot be $E$-CNR: indeed, by Lemma \ref{lem:CNR.NS} it would be $(E,m)$-NS. So the cube $\BC_{L_{k+1}}^{(n)}(\By)$ is $E$-CNR and $(E,m)$-S. This implies again by Lemma \ref{lem:CNR.NS} that
\[
M(\BC_{L_{k+1}}^{(n)}(\By),E)\geq J+1.
\]
Therefore $\omega\in\rS_{\By}$, so that $\rB_{k+1}\subset\Sigma\cup\rS_{\Bx}\cup\rS_{\By}$, hence
\[
\DP\left\{\rB_{k+1}\right\}
\leq\DP\{\Sigma\}+\DP\{\rS_{\Bx}\}+\DP\{\rS_{\By}\}
\]
and $\prob{\Sigma}\leq L_{k+1}^{-4^N\,p}$ by Theorem \ref{thm:Wegner}.
Now let us estimate $\DP\{\rS_{\Bx}\}$ and similarly $\DP\{\rS_{\By}\}$. Since 
\[
M_{\pai}(\BC_{L_{k+1}}^{(n)}(\Bx),E)+M_{\fui}(\BC_{L_{k+1}}^{(n)}(\Bx),E)\geq M(\BC_{L_{k+1}}^{(n)}(\Bx),E),
\] 
the inequality $M(\BC_{L_{k+1}}^{(n)}(\Bx),E)\geq\kappa(n)+ 6$, implies that either $M_{\pai}(\BC_{L_{k+1}}^{(n)}(\Bx),E)\geq\kappa(n)+ 2$, or $M_{\fui}(\BC_{L_{k+1}}^{(n)}(\Bx),E)\geq 4$. Therefore, by Lemma \ref{lem:MND} and Lemma \ref{lem:MD} (with $\ell=2$),
\begin{align*}
\DP\{\rS_{\Bx}\}&\leq\DP\left\{\exists\,E\in I:M_{\pai}(\BC_{L_{k+1}}^{(n)}(\Bx),E)\geq \kappa(n)+2\right\}\\
&\quad+\DP\left\{\exists\,E\in I:M_{\fui}(\BC_{L_{k+1}}^{(n)}(\Bx),E)\geq 4\right\}\\
&\leq\frac{3^{2nd}}{2}L_{k+1}^{2nd}(L_k^{-4^Np}+L_k^{-4p\,4^{N-n}})+C'(n,N,d)L_{k+1}^{4 dn-\frac{4 p}{\alpha}4^{N-n}}\\
&\leq C''(n,N,d)\left(L_{k+1}^{-\frac{4^Np}{\alpha}+2nd}+L_{k+1}^{-\frac{4p}{\alpha}4^{N-n}+2nd}+L_{k+1}^{-\frac{4p}{\alpha}4^{N-n}+4nd}\right)\\
&\leq C'''(n,N,d) L_{k+1}^{-\frac{4p}{\alpha}4^{N-n}+4nd}\tag{$\alpha=3/2$}\\
&\leq\frac{1}{4}L_{k+1}^{-2p\,4^{N-n}},
\end{align*}
where we used $p>4\alpha Nd=6Nd$. Finally 
\[ 
\prob{\rB_{k+1}}\leq L_{k+1}^{-4^Np}+\frac{1}{2}L_{k+1}^{-2p4^{N-n}}<L_{k+1}^{-2p4^{N-n}}.
\]
\end{proof}

\subsection{Mixed pairs of cubes}\label{ssec:mixed.S}
Finally, it remains only to derive $\dskonn$ in case (III), i.e., for pairs of $n$-particle cubes where one is PI while the other is FI.

\begin{theorem}\label{thm:mixed}
Let $1\leq n\leq N$. There exists $L_3^*=L_3^*(N,d)>0$ such that if $L_0\geq L_3^*(N,d)$ and if for $k\geq 0$,
\begin{enumerate}[\rm(i)]
\item
$\dsn{k-1,n'}$ holds true for all $1\leq n'<n$,
\item
$\dsn{k,n'}$ holds true for all $1\leq n'<n$ and
\item
$\dsknn$ holds true for all pairs of FI cubes,
\end{enumerate}
then $\dskonn$ holds true for any pair of separable cubes $\BC_{L_{k+1}}^{(n)}(\Bx)$, $\BC_{L_{k+1}}^{(n)}(\By)$ where one is PI while the other is FI.
\end{theorem}

\begin{proof}
Consider a pair of separable $n$-particle cubes $\BC_{L_{k+1}}^{(n)}(\Bx)$, $\BC_{L_{k+1}}^{(n)}(\By)$ and suppose that $\BC_{L_{k+1}}^{(n)}(\Bx)$ is PI while $\BC_{L_{k+1}}^{(n)}(\By)$ is FI. Set $J=\kappa(n)+5$ and introduce the events
\begin{align*}
\rB_{k+1}
&=\left\{\exists\,E\in I_0:\BC_{L_{k+1}}^{(n)}(\Bx)\text{ and $\BC_{L_{k+1}}^{(n)}(\By)$ are $(E,m)$-S}\right\},\\
\Sigma
&=\left\{\exists\,E\in I_0: \BC_{L_{k+1}}^{(n)}(\Bx)\ \text{is not $E$-CNR and}\ \BC_{L_{k+1}}^{(n)}(\By)\text{ is not $E$-CNR}\right\},\\
\mathcal{N}_{\Bx}
&=\left\{ \BC_{L_{k+1}}^{(n)}(\Bx)\text{ is $m$-non-localized}\right\},\\
\rS_{\By}
&=\left\{ \exists\,E\in I_0:M(\BC_{L_{k+1}}^{(n)}(\By),E)\geq J+1\right\}.
\end{align*}
Let $\omega\in \rB_{k+1}\setminus (\Sigma\cup \mathcal{N}_{\Bx})$, then for all $E\in I_0$ either $\BC_{L_{k+1}}^{(n)}(\Bx)$ is $E$-CNR or $\BC_{L_{k+1}}^{(n)}(\By)$ is $E$-CNR and  $\BC_{L_{k+1}}^{(n)}(\Bx)$ is $m$-localized. The cube $\BC_{L_{k+1}}^{(n)}(\Bx)$ cannot be $E$-CNR. Indeed, by Lemma \ref{lem:T.estimate} it would have been $(E,m)$-NS. Thus the cube $\BC_{L_{k+1}}^{(n)}(\By)$ is $E$-CNR, so by Lemma \ref{lem:CNR.NS}, $M(\BC_{L_{k+1}}^{(n)}(\By);E)\geq J+1$: otherwise $\BC_{L_{k+1}}^{(n)}(\By)$ would be $(E,m)$-NS. Therefore $\omega\in \rS_{\By}$. Consequently,
\[
\rB_{k+1}\subset \Sigma \cup \mathcal{N}_{\Bx}\cup\rS_{\By}.
\]
Recall that the probabilities $\DP\{\mathcal{N}_{\Bx}\}$ and $\DP\{\rS_{\By}\}$ have already been estimated in Sections \ref{ssec:PI.cubes} and \ref{ssec:FI.cubes}. We therefore obtain
\begin{align*}
\DP\left\{\rB_{k+1}\right\}&\leq \DP\{\Sigma\}+\DP\{\mathcal{N}_{\Bx}\}+\DP\{\rS_{\By}\}\\
&\leq L_{k+1}^{-4^Np}+ \frac{1}{2}L_{k+1}^{-4p\,4^{N-n}}+ \frac{1}{4}L_{k+1}^{-2p\,4^{N-n}}\leq L_{k+1}^{-2p\,4^{N-n}}.\qedhere
\end{align*}
\end{proof}

\section{Conclusion: the multi-particle multi-scale analysis}

\begin{theorem}\label{thm:DS.k.N}
Let $1\leq n\leq N$ and  $\BH^{(n)}(\omega)=-\BDelta+\sum_{j=1}^n V(x_j,\omega)+\BU$, where $\BU$, $V$ satisfy
 $\condI$ and $\condP$  respectively. There exists $m_n>0$ such that for any $p>6Nd$ property $\dsknn$ holds true for all $k\geq 0$ provided $L_0$ is large enough.
\end{theorem}
\begin{proof}
We prove that for each $n=1,\ldots,N$, property $\dsknn$ is valid. To do so, we use an induction on the number of particles $n'=1,\ldots,n$. For $n=1$  property $\dsn{k,1}$ holds true for all $k\geq 0$ by  the single-particle localization theory \cite{St01}. Now suppose that for all $1\leq n'<n$, $\dsn{k,n'}$ holds true for all	 $k\geq 0$, we aim to prove that $\dsknn$ holds true for all $k\geq 0$. For $k=0$, $\dsn{0,n}$ is valid using Theorem \ref{thm:np.initial.MSA}. Next, suppose that $\dsn{k',n}$ holds true for all $k'<k$, then by  combining this last assumption with $\dsn{k,n'}$ above, one can conclude that
\begin{enumerate}
\item[\rm(i)] $\dsknn$ holds true for all $k\geq 0$ and for all pairs of PI cubes using Theorem \ref{thm:partially.interactive},
\item[\rm(ii)] $\dsknn$ holds true for all $k\geq 0$ and for all pairs of FI cubes using Theorem \ref{thm:fully.interactive},
\item[\rm(iii)] $\dsknn$ holds true for all $k\geq 0$ and for all pairs of MI cubes using Theorem \ref{thm:mixed}.
\end{enumerate}
 Hence Theorem \ref{thm:DS.k.N} is proven.
\end{proof}

\section{Proofs of the results}

\subsection{Proof of Theorem \ref{thm:bottom.spectrum}}
Let $1\leq n\leq N$. We aim to prove $\sigma(\BH^{(n)}(\omega))=[0,+\infty)$ almost surely. Assumption $\condI$ implies that $\BU$ is non-negative and assumption $\condP$ also implies that $\BV$ is non-negative. Since, $-\BDelta\geq 0$, we get that almost surely $\sigma(\BH^{(n)}(\omega))\subset [0;+\infty)$. It remains to see that $[0;+\infty)\subset\sigma(\BH^{(n)}(\omega))$ almost surely. 

Let $k,m \in\DN$. Define,
\[
B_{k,m}:=\{\Bx\in\DZ^{nd}: \min_{i\neq j}|x_i-x_j|>r_0+2km\}
\]
where $r_0>0$, is the range of the interaction $\BU$. We also define the following sequence in $\DZ^{nd}$,
\[
\Bx^{k,m}:=C_{k,m}(1,\ldots,nd),
\]
where $C_{k,m}=r_0+2km+1$. Using the identification $\DZ^{nd}\cong (\DZ^d)^n$, we can also write $\Bx^{k,m}=C_{k,m}(x_1^{k,m},\ldots,x_n^{k,m})$ with each $x_i^{k,m}\in\DZ^d$, $i=1,\ldots,n$. Obviously, each term $\Bx^{k,m}$ of the sequence $(\Bx^{k,m})_{k,m}$ belongs to $B_{k,m}$. For $j=1,\ldots,n$, set,

\[
H^{(1)}_j(\omega):=-\Delta + V(x_j;\omega).
\]
We have that almost surely $\sigma(H^{(1)}_j(\omega))=[0;+\infty)$, see for example \cite{St01}. So, if we set for $j=1,\ldots, n$,
\[
\Omega_j:=\{\omega\in\Omega: \sigma(H^{(1)}_j(\omega))=[0,+\infty)\},
\]
$\prob{\Omega_j}=1$ for all $j=1,\ldots,n$. Now, put 
\[
\Omega_0:= \bigcap_{j=1}^n \Omega_j.
\]
We also have that $\prob{\Omega_0)}=1$. Let $\lambda\in[0;+\infty)$ and $\omega\in\Omega_0$, for this $\omega$, we have that almost surely, $\lambda\in\sigma(H^{(1)}(\omega))$ for all $j=1,\ldots,n$ and by the Weyl criterion, there exist $n$ Weyl sequences $\{(\phi^m_j)_m: j=1,\ldots,n\}$ related to $0$ and each operator $H^{(1)}_j(\omega)$. By the density property of compactly supported functions $C^{\infty}_c(\DR^d)$ in $L^2(\DR^d)$, we can directly assume that each $\phi_j^m$ is of compact support, i.e., $\supp \phi_j^m\subset C^{(1)}_{k_j m}(0)$ for some integer $k_j$ large enough. Set 
\[
k_0:=\max_{j=1,\ldots,n} k_j,
\]
and put, $\Bx^{k_0,m}=(x_1^{k_0,m},\ldots, x_n^{k_0,m})\in B_{k_0,m}$. We translate each function $\phi_j^m$ to have support contained in the $C^{(1)}_{k_0m}(x_j^{k_0})$. Next, consider the sequence $(\phi^m)_m$ defined by the tensor product,
\[
\phi^m:=\phi^m_1\otimes\cdots\otimes \phi^m_n. 
\]
We have that $\supp \phi^m\subset \BC^{(n)}_{k_0m}(\Bx^{k_0,m})$ and we aim to show that, $(\phi^m)_m$ is a Weyl sequence for $\BH^{(n)}(\omega)$ and $\lambda$. For any $\By\in\DR^{nd}$:

\[
|(\BH^{(n)}(\omega)\phi^m)(\By)|= |(\BH^{(n)}_0(\omega)\phi^m)|.
\]
Indeed, for the values of $\By$ inside the cube $\BC^{(n)}_{k_0m}(\Bx^{k_0,m})$ the interaction potential $\BU$ vanishes and for those values outside that cube, $\phi^m$ equals zero too. Therefore, 
\begin{align*}
\|\BH^{(n)}(\omega)\phi^m\|&\leq \| \BH^{(n)}_0(\omega)\phi^m\|\\
&\leq \sum_{j=1}^n \| (H^{(1)}_j(\omega)-\lambda)\phi_j^m\| \tto{m\to +\infty} 0,
\end{align*}
because, for all $j=1,\ldots,n$, $\|(H^{(1)}_j(\omega)-\lambda)\phi^m_j\|\rightarrow 0$ as $m\rightarrow +\infty$, since $\phi^m_j$is a Weyl sequence for $H^{(1)}_j(\omega)$ and $\lambda$. This completes the proof.

\subsection{Proof of Theorem \ref{thm:exp.loc}}
		
Using the multi-particle multi-scale analysis bounds in the continuum property $\dsn{k,N}$, we extend to multi-particle systems the strategy of Stollmann \cite{St01}.	

For $\Bx_0\in\DZ^{Nd}$ and an integer $k\geq 0$, set, using the notations of Lemma \ref{lem:separable.distant}
\[
R(\Bx_0):=\max_{1\leq \ell\leq \kappa(N)}|\Bx_0-\Bx^{(\ell)}|;\qquad  b_{k}(\Bx_0):=7N+R(\Bx_0)L_k^{-1},
\]
\[
M_k(\Bx_0):=\bigcup_{\ell=1}^{\kappa(N)}\BC^{(N)}_{7NL_k}(\Bx^{(\ell)}
\] 
and define
\[
A_{k+1}(\Bx_0):=\BC^{(N)}_{bb_{k+1}L_{k+1}}(\Bx_0)\setminus \BC^{(N)}_{b_kL_k}(\Bx_0),
\]
where the parameter $b>0$ is to be chosen later. We can easily check that,
\[
M_k(\Bx_0)\subset \BC^{(N)}_{b_kL_k}(\Bx_0).
\]
Moreover, if $\Bx\in A_{k+1}(\Bx_0)$, then the cubes $\BC^{(N)}_{L_k}(\Bx)$ and $\BC^{(N)}_{L_k}(\Bx_0)$ are separable by Lemma \ref{lem:separable.distant}. Now, also define
\[
\Omega_k(\Bx_0):=\{\text{$\exists E\in I_0$ and $\Bx\in A_{k+1}(\Bx_0)\cap \Gamma_k$: $\BC^{(N)}_{L_k}(\Bx)$ and $\BC^{(N)}_{L_k}(\Bx_0)$ are $(E,m)$-S}\},
\]
with $\Gamma_k:= \Bx_0+\frac{L_k}{3}\DZ^{Nd}$. Now, property $\dsn{k,N}$ combined with the cardinality of $A_{k+1}(\Bx_0)\cap \Gamma_k$ imply 
\begin{align*}
\prob{\Omega_k(\Bx_0)}&\leq (2bb_{k+1}L_{k+1})^{Nd}L_k^{-2p}\\
&\leq (2bb_{k+1})^{Nd}L^{-2p+\alpha Nd}. 
\end{align*}
Since, $p>(\alpha Nd+1)/2$ (in fact, $p>6Nd$), we get 
\[
\sum_{k=0}^{\infty}\prob{\Omega_k(\Bx_0)}<\infty.
\]
Thus, setting 
\[
\Omega_{<\infty}:=\{\text{ $\forall \Bx_0\in\DZ^{Nd}$, $\Omega_k(\Bx_0)$ occurs finitely many times} \},
\]
by the Borel cantelli Lemma and the countability of $\DZ^{Nd}$ we have that $\prob{\Omega_{<\infty}}=1$. Therefore it suffices to pick $\omega\in \Omega_{<\infty}$ and prove the exponential decay of any nonzero eigenfunction $\BPsi$ of $\BH^{(N)}(\omega)$.

Let $\BPsi$ be a polynomially bounded  eigenfunction satisfying (EDI) (see Theorem \ref{thm:GRI.EF}). Let $\Bx_0\in\DZ^{Nd}$ with $\|\Bone_{\BC^{(N)}_1(\Bx_0)}\BPsi\|>0$ (if there is no such $\Bx_0$, we are done). The cube $\BC^{(N)}_{L_k}(\Bx_0)$ cannot be $(E,m)$-NS for infinitely many $k$. Indeed, given an integer $k\geq0$, if $\BC^{(N)}_{L_k}(\Bx_0)$ is $(E,m)$-NS  then by (EDI), and the polynomial bound on $\BPsi$, we get
\begin{align*}
\|\Bone_{\BC^{(N)}_1(\Bx_0)}\BPsi\|&\leq C\cdot\|\Bone_{\BC^{(N,out)}_{L_k}(\Bx_0)}\BG^{(N)}_{\BC^{(N)}_{L_k}(\Bx_0}(E)\Bone_{\BC^{(N,int)}_{L_k}(\Bx_0}\|\cdot\|\Bone_{\BC^{(N,out)}_{L_k}(\Bx_0)}\BPsi\|\\
&\leq C(1+|\Bx_0|+L_k)^t\cdot \ee^{-mL_k}\tto{L_k\rightarrow \infty} 0,
\end{align*}
in contradiction with the choice of $\Bx_0$. So there is an integer $k_1=k_1((\omega,E,\Bx_0)<\infty$ such that $\forall k\geq k_1$ the cube $\BC^{(N)}_{L_k}(\Bx_0)$ is $(E,m)$-S.  At the same time, since $\omega\in \Omega_{<\infty}$, there exists $k_2=k_2(\omega,\Bx_0)$ such that if $k\geq k_2$, $\Omega_k(\Bx_0)$ does not occurs. We conclude that for all $k\geq \max\{k_1,k_2\}$, for all $\Bx\in A_{k+1}(\Bx_0)\cap\Gamma_k$, $\BC^{(N)}_{L_k}(\Bx)$ is $(E,m)$-NS. 

Let $\rho\in(0,1)$ and choose $b>0$ such that
 \[
b>\frac{1+\rho}{1-\rho},
\]
so that
\[
\BC^{(N)}_{\frac{bb_{k+1}L_{k+1}}{1-\rho}}(\Bx_0)\setminus\BC^{(N)}_{\frac{b_kL_k}{1-\rho}}(\Bx_0)\subset A_{k+1}(\Bx_0),
\]
for $\Bx\in\tilde{A}_{k+1}(\Bx_0)$.
\begin{enumerate}
\item[(1)]
Since, $|\Bx-\Bx_0|>\frac{b_kL_k}{1-\rho}$,
\begin{align*}
\dist(\Bx,\partial\BC^{(N)}_{b_kL_k}(\Bx_0)&\geq |\Bx-\Bx_0|-b_kL_k\\
&\geq |\Bx-\Bx_0|-(1-\rho)|\Bx-\Bx_0|\\
&=\rho(|\Bx-\Bx_0|)
\end{align*}                              

\item[(2)]
Since $|\Bx-\Bx_0|\leq \frac{bb_{k+1}L_{k+1}}{1+\rho}$,
\begin{align*}
\dist(\Bx,\partial\BC^{(N)}_{bb_{k+1}L_{k+1}}(\Bx_0))&\geq bb_{k+1}L_{k+1}-|\Bx-\Bx_0|\\
&\geq (1+\rho)|\Bx-\Bx_0|-|\Bx-\Bx_0|\\
&=\rho|\Bx-\Bx_0|.
\end{align*}
\end{enumerate}
Thus, 
\[
\dist(\Bx,\partial A_{k+1}(\Bx_0))\geq \rho|\Bx-\Bx_0|.
\]
Now, setting $k_3=\max\{k_1,k_2\}$, the assumption linking $b$ and $\rho$ implies that:
\[
\bigcup_{k\geq k_3} \tilde{A}_{k+1}(\Bx_0)=\DR^{Nd}\setminus\BC^{(N)}_{\frac{b_{k_3}L_{k_3}}{1-\rho}}(\Bx_0),
\]
because $\frac{bb_{k+1}L_{k+1}}{1+\rho}>\frac{b_kL_k}{1-\rho}$. Let $k\geq k_3$, recall that this implis that all the cubes with centers in $A_{k+1}(\Bx_0)\cap\Gamma_k$ and side length $2L_k$ are $(E,m)$-NS. Thus, for any $\Bx\in\tilde{A}_{k+1}(\Bx_0)$, we choose $\Bx_1\in A_{k+1}(\Bx_0)$ such that $\Bx\in \BC^{(N,int)}_{L_k}(\Bx_1)$. Therefore
\begin{align*}
\|\Bone_{\BC^{(N)}_1(\Bx)}\BPsi\|&\leq \|\Bone_{\BC^{(N,int)}_{L_k}(\Bx_1)}\BPsi\|\\
&\leq C\cdot\ee^{-mL_k}\cdot\|\Bone_{\BC^{(N,out)}_{L_k}(\Bx_1)}\BPsi\|
\end{align*}

Up to a set of Lebesgue measure zero, we can cover $\BC^{(N,out)}_{L_k}(\Bx_1)$ by at most $3^{Nd}$ cubes 
\[
\BC^{(N,int)}_{L_k}(\tilde{\Bx}),\qquad \tilde{\Bx}\in\Gamma_k,\quad |\tilde{\Bx}-\Bx_1|=\frac{L_k}{3}.
\]
By choosing $\Bx_2$ which gives a maximal norm, we get
\[
\|\Bone_{\BC^{(N,out)}_{L_k}(\Bx_1)}\BPsi\|\leq 3^{Nd}\cdot\|\Bone_{\BC^{(N,int)}_{L_k}(\Bx_2)}\BPsi\|,
\]
so that 
\[
\|\Bone_{\BC^{(N)}_1(\Bx)}\BPsi\|\leq  3^{Nd}\cdot\ee^{-m L_k}\cdot\|\Bone_{\BC^{(N,int)}_{L_k}(\Bx_2)}\BPsi\|.
\]
Thus, by an induction procedure, we find a sequence $\Bx_1$, $\Bx_2$, ..., $\Bx_n$ in $\Gamma_k\cap A_{k+1}(\Bx_0)$  and the bound
\[
\|\Bone_{\BC^{(N)}_1(\Bx)}\BPsi\|\leq (C\cdot 3^{Nd}\exp(-mL_k))^n\cdot\|\Bone{\BC^{(N,out)}_{L_k}(\Bx_n)}\BPsi\|.
\]
Since $|\Bx_i-\Bx_{i+1}|=L_k/3$ and  $\dist(\Bx,\partial A_{k+1})\geq \rho\cdot|\Bx-\Bx_0|$, we can iterate at least $\rho\cdot|\Bx-\Bx_0|\cdot3/L_k$ times until, we reach the boundary of $A_{k+1}(\Bx_0)$. Next, using the polynomial bound on $\BPsi$, we obtain:

\begin{gather*}
\|\Bone_{\BC^{(N)}_1(\Bx)}\BPsi\|\leq (C\cdot 3^{Nd})^{\frac{3\rho|\Bx-\Bx_0|}{L_k}}\cdot\exp(-3m\rho|\Bx-\Bx_0|)\\
\times C(1+|\Bx_0|+ bL_{k+1})^t\cdot L_{k+1}^{Nd}.
\end{gather*}

We can conclude that given $\rho'$ with $0<\rho'<1$, we can find $k_4\geq k_3$ such that if $k\geq k_4$, then
\[
\|\Bone_{\BC^{(N)}_1(\Bx)}\BPsi\|\leq \ee^{-\rho\rho'm|\Bx-\Bx_0|},
\]
if $|\Bx-\Bx_0|>\frac{b_{k_4}L_{k_4}}{1-\rho}$. This completes the proof of the exponential localization in the max-norm.

\subsection{Proof of Theorem \ref{thm:dynamical.loc}}
		
For the proof of the multi-particle dynamical localization given  the multi-particle multi-scale analysis in the continuum, we refer to the paper by Boutet de Monvel et al. \cite{BCS11}.

\section{Appendix}\label{sec:appendix}
\begin{center}

\end{center}

\subsection{Proof of Lemma \ref{lem:separable.distant}}

(A) Let $L>0$, $\emptyset \neq \CJ\subset \{1,\ldots,n\}$ and $\By\in\DZ^{nd}$. $\{y_j\}_{j\in\CJ}$ is called an $L$-cluster if the union
\[
\bigcup_{j\in\CJ} C^{(1)}_L(y_j)
\]
cannot be decomposed into two non-empty disjoint subsets. Next, given two configurations $\Bx,\By\in\DZ^{nd}$, we proceed as follows:
\begin{enumerate}
\item
We decompose the vector $\By$ into maximal $L$-clusters $\Gamma_1,\ldots,\Gamma_M$ (each of diameter $\leq 2nL$) with $M\leq n$.
\item
Each position $y_i$ corresponds to exactly one cluster $\Gamma_j$, $j=j(i)\in\{1,\ldots,M\}$.
\item
If there exists $j\in\{1,\ldots,M\}$ such that $\Gamma_j\cap\varPi\BC^{(n)}_L(\Bx)=\emptyset$, then cubes $\BC^{(n)}_L(\By)$ and $\BC^{(n)}_L(\Bx)$ are separable.
\item
If (3)  is wrong, then for all $k=1,\ldots,M$, $\Gamma_k\cap \varPi\BC^{(n)}_L(\Bx)\neq \emptyset$. Thus for all $k=1,\ldots,M$, $\exists i=1,\ldots,n$ such that $\Gamma_k\cap \BC^{(1)}_L(x_i)\neq \emptyset$. Now for any $j=1,\ldots,n$ there exists $k=1,\ldots,M$ such that $y_j\in\Gamma_k$. Therefore for such $k$, by hypothesis there exists $i=1,\ldots,n$ such that $\Gamma_k\cap C^{(1)}_L(x_i)\neq \emptyset$. Next let $z\in\Gamma_k\cap C^{(1)}_L(x_i)$ so that $|z-x_i|\leq L$. We have that
\begin{align*}
|y_j-x_i|&\leq |y_j-z|+|z-x_i|\\
&\leq 2nL-L + L=2nL
\end{align*}
\end{enumerate}
since $y_j,z\in\Gamma_k$. Notice that above we have the bound $|y_j-z|\leq 2nL-L$ instead of $2nL$ because $y_j$ is a center of the $L$-cluster $\Gamma_k$. Hence for all $j=1,\ldots,n$ $y_j$ must belong to one of the cubes $C^{(1)}_{2nL}(x_i)$ for the $n$ positions $(y_1,\ldots,y_n)$. Set $\kappa(n)=n^n$. For any choice of at most $\kappa(n)$ possibilities, $\By=(y_1,\ldots,y_n)$ must belong to the Cartesian product of $n$ cubes of size $2nL$ i.e. to an $nd$-dimensional  cube of size  $2nL$, the assertion then follows.

(B) Set $R(\By)=\max_{1\leq i,j\leq n}|y_i-y_j|+5NL$ and consider a cube $\BC^{(n)}_L(\Bx)$ with $|\By-\Bx|> R(\By)$. Then there exists $i_0\in\{1,\ldots,n\}$ such that $|y_{i_0}-x_{i_0}|>R(\By)$. Consider the maximal connected component $\Lambda_{\Bx}:=\bigcup_{i\in\CJ} C^{(1)}_{L}(x_i)$ of the union $\bigcup_i C^{(1)}_L(x_i)$ containing $x_{i_0}$. Its diameter is bounded by $2nL$. We have
\[
\dist(\Lambda_{\Bx},\varPi\BC^{(n)}_L(\By))=\min_{u,v}|u-v|,
\]
now since
\[
|x_{i_0}-y_{i_0}|\leq |x_{i_0}-u|+|u-v|+|v-y_{i_0}|,
\]
then

\begin{align*}
\dist(\Lambda_{\Bx},\varPi\BC^{(n)}_L(\By))&=\min_{u,v}|u-v|\\
&\geq |x_{i_0}-y_{i_0}|-\diam(\Lambda_{\Bx})-\max_{v,y_{i_0}}|v-y_{i_0}|.\\
\end{align*}
Recall that $\diam(\Lambda_{\Bx})\leq 2nL$ and
\[
\max_{v,y_{i_0}}|v-y_{i_0}|\leq \max_{v}|v-y_j|+\max_{y_{i_0}}|y_j-y_{i_0}|,
\]
for some $j=1,\cdots n $ such that $v\in C^{(1)}_L(y_j)$. Finally we get
\[
\dist(\Lambda_{\Bx},\varPi\BC^{(n)}_L(\By))>R(\By) -\diam(\Lambda_{\Bx})-(2L + \diam(\varPi\By))>0,
\]
this implies that $\BC^{(n)}_L(\Bx)$ is $\CJ$-separable from $\BC^{(n)}_L(\By)$ with $\CJ$ the index subset appearing in the definition of $\Lambda_{\Bx}$.

\subsection{Proof of Lemma \ref{lem:PI.cubes}}
It is convenient to use the canonical injection $\DZ^d\hookrightarrow\DR^d$; then the notion of connectedness in $\DR^d$ induces its analog for lattice cubes. Set $R:=2L+r_0$ and assume that $\diam\varPi\Bu = \max_{i,j}|u_i - u_j|>nR$.
If the union of cubes $C_{R/2}^{(1)}(u_i)$, $1\leq i\leq n$, were not decomposable into two (or more) disjoint groups, then it would be connected, hence its diameter would be bounded by $n(2(R/2))=nR$, hence $\diam \varPi\Bu\leq nR$ which contradicts the hypothesis. Therefore, there exists an index subset $\CJ\subset\{1,\ldots,n\}$ such that $|u_{j_1}-u_{j_2}|>2(R/2)$ for all $j_1\in\CJ$, $j_2\in\CJ^c$, this implies that
\begin{align*}
\dist\left(\varPi_{\CJ}\BC_L^{(n)}(\Bu),\varPi_{\CJ^{\comp}}\BC_L^{(n)}(\Bu)\right)&=\min_{j_1\in\CJ,j_2\in\CJ^c}\dist\left(C^{(1)}_L(u_{j_1}),C^{(1)}_L(u_{j_2})\right)\\
&\geq \min_{j_1\in\CJ,j_2\in\CJ^{c}}|u_{j_1}-u_{j_2}|-2L>r_0.
\end{align*}

\subsection{Proof of Lemma \ref{lem:FI.cubes}}

If for some $R>0$,
\[
R<|\Bx-\By|=\max_{1\leq j\leq n}|x_j-y_j|,
\]
then there exists $1\leq j_0\leq n$ such that $|x_{j_0}-y_{j_0}|>R$. Since both cubes are fully interactive, by Definition \eqref{def:diagonal.cubes}
\begin{align*}
&|x_{j_0}-x_i| \le \diam \varPi \Bx \le n(2L+r_0),\\
&|y_{j_0}-y_j| \le \diam \varPi \By \le n(2L+r_0).
\end{align*}
By the triangle inequality, for any $1\leq i,j\leq n$ and $R>7\,nL >6\,nL+2\,nr_0$, we have
\begin{align*}
|x_i-y_j|&\geq |x_{j_0}-y_{j_0}|-|x_{j_0}-x_i|-|y_{j_0}-y_j|\\
&>6nL+2nr_0-2n(2L+r_0)=2nL.
\end{align*}
Therefore, for any $1\leq i,j\leq n$,
\[
\min_{i,j}\dist(C^{(1)}_L(x_i),C^{(1)}_L(y_j))\geq \min_{i,j}|x_i-y_j|-2L>2(n-1)L\geq 0,
\]
which proves the claim.

\subsection{Proof of Lemma \ref{lem:MPI}}
Assume that $M^{\sep}(\BC^{(n)}_{L_{k+1}}(\Bu),E)<2$, (i.e., there is no pair of separable  cubes of radius $L_k$ in $\BC^{(n)}_{L_{k+1}}(\Bu)$), but $M(\BC^{(n)}_{L_{k+1}}(\Bu),E)\geq \kappa(n)+2$. Then $\BC^{(n)}_{L_{k+1}}(\Bu)$ must contain at least $\kappa(n)+2$ cubes $\BC^{(n)}_{L_k}(\Bv_i)$, $0\leq i\leq \kappa(n)+1$ which are non separable but satisfy $|\Bv_i-\Bv_{i'}|>7NL_k$, for all $i\neq i'$. On the other hand, by Lemma \ref{lem:separable.distant} there are at most $\kappa(n)$ cubes $\BC^{(n)}_{2nL_k}(\By_i)$, such that any cube $\BC^{(n)}_{L_k}(\Bx)$ with $\Bx\notin \bigcup_{j} \BC^{(n)}_{2nL_k}(\By_j)$ is separable from $\BC^{(n)}_{L_k}(\Bv_0)$. Hence $\Bv_i\in \bigcup_{j} \BC^{(n)}_{2nL_k}(\By_j)$ for all $i=1,\ldots,\kappa(n)+1$.  But since for all $i\neq i'$, $|\Bv_i-\Bv_{i'}|>7NL_k$, there must be at most one center $\Bv_i$ per cube $\BC^{(n)}_{2nL_k}(\By_j)$, $1\leq j\leq \kappa(n)$. Hence we come to a contradiction:
\[
\kappa(n)+1\leq \kappa(n).
\]
 The same analysis holds true if we consider only PI cubes.

\subsection{Proof of Lemma \ref{lem:MND}}
Suppose that $M_{\pai}(\BC^{(n)}_{L_{k+1}}(\Bu),I)\geq \kappa(n)+2$, then by Lemma \ref{lem:MPI}, $M_{\pai}^{\sep}(\BC^{(n)}_{L_{k+1}}(\Bu), I)\geq 2$, i.e., there are at least two separable $(E,m)$-singular PI cubes $\BC^{(n)}_{L_k}(\Bu^{(j_1)})$, $\BC^{(n)}_{L_k}(\Bu^{(j_2)})$ inside $\BC^{(n)}_{L_{k+1}}(\Bu)$.
The  number of possible pairs of centers $\{\Bu^{(j_1)},\Bu^{(j_2)}\}$ such that
\[
\BC_{L_k}^{(n)}(\Bu^{(j_1)}),\,\BC_{L_k}^{(n)}(\Bu^{(j_2)})\subset\BC_{L_{k+1}}^{(n)}(\Bu)
\]
 is bounded by $\frac{3^{2nd}}{2}L_{k+1}^{2nd}$. Then, setting 
\[
\rB_k=\{\text{$\exists E\in I$, $\BC^{(n)}_{L_k}(\Bu^{(j_1)})$, $\BC^{(n)}_{L_k}(\Bu^{(j_2)})$ are $(E,m)$-S}\},
\]
\[
\DP\left\{M_{\pai}^{\sep}(\BC_{L_{k+1}}^{(n)}(\Bu),I)\geq 2\right\}\leq\frac{3^{2nd}}{2}L_{k+1}^{2nd}\times\prob{\rB_k}
\]
with  $\prob{\rB_k}\leq L_k^{-4^Np}+L_k^{-4p\,4^{N-n}}$ by \eqref{eq:bound.PI}.
Here $\rB_k$ is defined as in Theorem \ref{thm:partially.interactive}.

\subsection{Proof of Lemma \ref{lem:MD}}

\begin{proof}
Suppose there exist $2\ell$ pairwise separable, fully interactive cubes $\BC_{L_k}^{(n)}(\Bu^{(j)})$ $\subset\BC_{L_{k+1}}^{(n)}(\Bu)$, $1\leq j\leq 2\ell$. Then, by Lemma \ref{lem:FI.cubes}, for any pair $\BC_{L_k}^{(n)}(\Bu^{(2i-1)})$, $\BC_{L_k}^{(n)}(\Bu^{(2i)})$, the corresponding random Hamiltonians $\BH_{\BC_{L_k}^{(n)}(\Bu^{(2i-1)})}^{(n)}$ and $\BH^{(n)}_{\BC_{L_k}^{(n)}(\Bu^{(2i)})}$ are independent, and so are their spectra and their Green functions. For $i=1,\dots,\ell$ we consider the events:
\[
\rA_i=\left\{\exists\,E\in I:\BC_{L_k}^{(n)}(\Bu^{(2i-1)})\text{ and $\BC_{L_k}^{(n)}(\Bu^{(2i)})$ are $(E,m)$-S}\right\}.
\]
Then by assumption $\dsknn$, we have, for $i=1,\dots,\ell$,
\begin{equation}
\DP\left\{\rA_i\right\}\leq L_k^{-2p\,4^{N-n}},
\end{equation}
and, by independence of events $\rA_1,\dots,\rA_{\ell}$,
\begin{equation}
\DP\Bigl\{\bigcap_{1\leq i\leq\ell}\rA_i\Bigr\}=\prod_{i=1}^{\ell}\DP(\rA_i)\leq\bigl(L_k^{-2p\,4^{N-n}}\bigr)^{\ell}.
\end{equation}
To complete the proof, note that the total number of different families of $2\ell$ cubes $\BC_{L_k}^{(n)}(\Bu^{(j)})\subset\BC_{L_{k+1}}^{(n)}(\Bu)$, $1\leq j\leq 2\ell$, is bounded by
\[
\frac{1}{(2\ell)!}\left|\BC_{L_{k+1}}^{(n)}(\Bu)\right|^{2\ell}\leq C(n,N,\ell,d)L_{k}^{2\ell dn\alpha}.\qedhere
\]
\end{proof}

\end{document}